\documentclass[11pt]{article}

\usepackage{amsmath, amsthm, amssymb, graphicx, enumerate, fullpage,amsfonts,mathrsfs,mathpazo,xspace}
\usepackage{color}
\usepackage{verbatim}
\usepackage{hyperref}
\usepackage{dsfont}

\usepackage[font={small}]{caption}
\usepackage{enumitem}
\usepackage{algpseudocode} 
\usepackage{algorithm}
\usepackage{algorithmicx}
\usepackage{float}

\usepackage{pgf}
\usepackage{tikz}
\usepackage[utf8]{inputenc}
\usetikzlibrary{arrows,automata}
\usetikzlibrary{positioning}
\usetikzlibrary{decorations.pathreplacing}

\makeatletter
\DeclareCaptionLabelFormat{protocol}{}
\makeatother

\theoremstyle{plain}
\newtheorem{thm}{Theorem}[section]
\newtheorem{theorem}[thm]{Theorem}
\newtheorem{definition}[thm]{Definition}
\newtheorem{proposition}[thm]{Proposition}

\newtheorem{lemma}[thm]{Lemma}

\newtheorem{fact}[thm]{Fact}

\theoremstyle{definition}

\newcommand{\ket}[1]{|#1 \rangle}

\newcommand{\braket}[2]{\langle #1|#2 \rangle}
\newcommand{\ketbra}[2]{|#1 \rangle\!\langle #2 |}
\newcommand{\tr}{\mathrm{tr}}
\newcommand{\ident}{\mathds{1}}

\newcommand{\MinEntropy}{H_\mathrm{min}}
\newcommand{\E}{\mathbb{E}}

\newcommand{\N}{\mathbb{N}}

\newcommand{\R}{\mathbb{R}}

\newcommand{\eps}{\varepsilon}

\definecolor{mygrey}{gray}{0.50}

\newcommand*{\etal}{et al.\@\xspace}
\newcommand*{\etc}{etc.\@\xspace}
\newcommand*{\eg}{e.g.\@\xspace}
\newcommand*{\ie}{i.e.\@\xspace}

\newcommand{\Hilb}{\mathcal{H}}
\newcommand{\VV}{\mathsf{VV}}
\newcommand{\RUV}{\mathsf{RUV}}
\newcommand{\EC}{\mathsf{EC}}
\newcommand{\CHSHN}{\mathsf{CHSH}^{\otimes N}}
\newcommand{\CHSH}[1]{\mathsf{CHSH}^{\otimes #1}}
\newcommand{\WIN}{\mathrm{WIN}}
\newcommand{\InfiniteExp}{\mathsf{InfiniteExpansion}}

\newcommand{\ClusterExp}{\mathsf{ClusterExpansion}}
\newcommand{\QExt}{\mathsf{QExt}}
\newcommand{\Fidelity}{\mathcal{F}}
	
\begin{document}

\title{Infinite Randomness Expansion and Amplification with a Constant Number of Devices}
\author{
\makebox[.2\columnwidth]{Matthew Coudron\thanks{Computer Science and Artificial Intelligence Laboratory, Massachusetts Institute of Technology. H.Y. was supported by an NSF Graduate Fellowship Grant No. 1122374 and National Science Foundation Grant No. 1218547.  M.C. was supported by the National Science Foundation under Grant No. 0801525.}} \\
\texttt{\href{mailto:mcoudron@mit.edu}{\color{black}mcoudron@mit.edu}} \\
MIT CSAIL
\and 
\makebox[.2\columnwidth]{Henry Yuen\footnotemark[1]} \\
\texttt{\href{mailto:hyuen@csail.mit.edu}{\color{black}hyuen@csail.mit.edu}} \\
MIT CSAIL
}

\maketitle

\begin{abstract}

We present a device-independent randomness expansion protocol, involving only a constant number of non-signaling quantum devices, that achieves \emph{infinite expansion}: starting with $m$ bits of uniform private randomness, the protocol can produce an unbounded amount of certified randomness that is $\exp(-\Omega(m^{1/3}))$-close to uniform and secure against a quantum adversary. The only parameters which depend on the size of the input are the soundness of the protocol and the security of the output (both are inverse exponential in $m$).  This settles a long-standing open problem in the area of randomness expansion and device-independence.  


The analysis of our protocols involves overcoming fundamental challenges in the study of \emph{adaptive} device-independent protocols. 
Our primary technical contribution is the design and analysis of device-independent protocols which are \emph{Input Secure}; that is, their output is guaranteed to be secure against a quantum eavesdropper, \emph{even if the input randomness was generated by that same eavesdropper}! 

The notion of Input Security may be of independent interest to other areas such as device-independent quantum key distribution.

\end{abstract}

\section{Introduction}\label{sec:intro}

Bell's Theorem states that the outcomes of local measurements on spatially separated systems cannot be predetermined, due to the phenomenon of quantum entanglement~\cite{bell1964einstein}. This is one of the most important ``no-go'' results in physics because it rules out the possibility of a local hidden variable theory that reproduces the predictions of quantum mechanics. However, Bell's Theorem has also found application in quantum information as a \emph{positive} result, in that it gives a way to certify the generation of genuine randomness: if measurement outcomes of separated systems exhibit non-local correlations (\eg correlations that violate so-called Bell Inequalities), then the outcomes cannot be deterministic. 

While Bell's Theorem does give a method to certify randomness, there is a caveat.  The measurement settings used on the separated systems have to be chosen at random! Nevertheless, it is possible to choose the measurement settings in a randomness-efficient manner such that the measurement outcomes certifiably contain \emph{more} randomness (as measured by, say, min-entropy) than the amount of randomness used as input. This is the idea behind \emph{randomness expansion protocols}, in which a classical experimenter, starting with $m$-bits of uniform randomness, can interact with physically isolated devices to certifiably generate $g(m)$ bits of (information theoretic) randomness  (ideally with $g(m) \gg m$). Furthermore, these protocols are \emph{device-independent}: the only assumption made on the devices is that they cannot communicate, and obey the laws of quantum mechanics. In particular, there is no \emph{a priori} assumption on the internal structure or dynamics of the devices.  Indeed, the devices may even have been manufactured by an adversary!

First proposed by Colbeck~\cite{Colbeck2009} in 2006, device-independent randomness expansion has flourished into an active area of research~\cite{Colbeck2011, Pironio2010, Vazirani2012, Fehr2013, Coudron2013, acin2012randomness, um2013experimental, gallego2013device, ms14}. Its study has synthesized a diverse array of concepts from quantum information theory, theoretical computer science, and quantum cryptography, including generalized Bell inequalities~\cite{Pironio2010, acin2012randomness, pironio2013security, Fehr2013}, the monogamy of entanglement~\cite{Vazirani2012, Reichardt2012}, randomness extractors~\cite{renner2008security, konig2008bounded,de2012trevisan}, and quantum key distribution~\cite{barrett2005no,masanes2011secure, vazirani2012fully, ms14}. Randomness expansion has even been experimentally realized by~\cite{Pironio2010}, who reported the generation of 42 bits of certified randomness (over the course of a month).

The fundamental problem in analyzing a randomness expansion protocol is in demonstrating a lower bound on the amount of certified randomness, usually measured by min-entropy. There have been a couple of different approaches. A line of works, starting with~\cite{Pironio2010}, gives bounds on the min-entropy by analytically relating the extent to which a Bell inequality is violated to the ``guessing probability'' of the protocol's output~\cite{Pironio2010, Fehr2013, acin2012randomness,pironio2013security}. Another approach, developed in \cite{Vazirani2012}, is to utilize the operational definition of min-entropy in a ``guessing game", which establishes that a low min-entropy output implies that the non-signaling devices must have communicated during the protocol (a contradiction). This latter approach yields a protocol (which we will refer to as the Vazirani-Vidick protocol in this paper) that not only achieves the state-of-the-art expansion factor $g(m) = \exp(m^{1/3})$, but is also \emph{quantum secure}: that is, the output contains high min-entropy even from the perspective of a malicious eavesdropper that may be entangled with the protocol devices. Recently, a work by~\cite{ms14} not only achieves quantum security, but randomness expansion that tolerates a constant level of noise in the devices.

The original protocol of~\cite{Colbeck2009,Colbeck2011} obtained $g(m) = \Theta(m)$, or linear expansion. This was improved by Pironio \etal~\cite{Pironio2010} to achieve quadratic expansion $g(m) = \Theta(m^2)$. The protocols of~\cite{Vazirani2012,Fehr2013,ms14} achieve exponential expansion. Perhaps the most tantalizing open question in randomness expansion is: how large an expansion factor $g(m)$ can we achieve? For example, is there a protocol with expansion factor $g(m)$ that is doubly-exponential in $m$? Is there any upper bound on randomness expansion in general? 

The only known upper bounds on randomness expansion apply to \emph{non-adaptive} protocols with two devices (\ie, where the referee's inputs to the devices do not depend on their previous outputs)~\cite{Coudron2013}. There the authors showed that \emph{noise robust}, non-adaptive protocols must have a finite bound on their expansion factor\footnote{They showed that $g(m) \leq \exp(\exp(m))$, or a doubly-exponential upper bound.}. With the exception of~\cite{Fehr2013}, randomness expansion protocols prior to our work were \emph{non-adaptive}, and hence the results of~\cite{Coudron2013} suggest those protocols have a bounded expansion factor. Thus, going beyond the the finite expansion barrier appears to require adaptivity -- but it could, \emph{a priori}, be the case that even adaptive protocols are inherently limited to finite randomness expansion.

\textbf{We present an adaptive protocol that achieves \emph{infinite} certifiable randomness expansion, using a \emph{constant} number of non-signaling quantum devices.}  The output length of our protocol depends only on the number of rounds performed in the protocol (which can be arbitrarily large), and not on the size of the initial random seed!  This shows that there is no finite upper bound on the expansion factor of adaptive protocols.   Our protocol involves a constant number -- eight, specifically -- of non-communicating black-box quantum devices, and guarantees that the output of the protocol is close to uniformly random, even from the point of view of a quantum eavesdropper (where the closeness to uniformity is determined by the initial seed length). Our protocol works even in the presence of arbitrary entanglement between the devices and an eavesdropper. 

The key technical component of the analysis of the $\InfiniteExp$ protocol is to show that a sub-protocol, which we call $\ClusterExp$, is \emph{Input Secure}: it generates uniform randomness secure against a quantum adversary, \emph{even if that adversary generated the seed randomness earlier in the protocol}! Since the $\ClusterExp$ sub-protocol is Input Secure, composing $\ClusterExp$ with itself in sequence (\ie using the outputs of one instance of the protocol as the inputs of another instance) yields another randomness expansion protocol, this time with much larger expansion factor. Our $\InfiniteExp$ protocol is the infinite composition of the $\ClusterExp$ sub-protocol. 

In Section~\ref{sec:related}, we discuss two relevant and enlightening results about randomness expansion~\cite{csw14,ms14}, which were announced after the original posting of this work (though these results were discovered independently and, unbeknownst to the authors, developed in parallel with this work). 

We note here that any exponential randomness expansion protocol with security against a quantum eavesdropper (such as the Vazirani-Vidick protocol, for example) readily yields a protocol using $2N$ devices, which has a randomness expansion given by an exponential tower function of $N$ (i.e. $2^{2^{2^{{...^2}^{N}}}}$): after running such a quantum-secure expansion protocol on one pair of devices, the devices are discarded, and their outputs are fed into a fresh pair of devices (that did not communicate with any previous devices used in the protocol). This ``exponential tower'' protocol terminates when all $2N$ devices have been used. This was first observed by~\cite{yuen13}, and in \cite{ms14} it is noted that the robust exponential expansion protocol given therein can be used to obtain an analogous ``tower'' randomness expansion protocol, which is also \textit{robust}. 

For all practical intents and purposes, a ``tower'' expansion protocol can certify much, much (... $\text{much}^{\text{much}^{\text{much}^{...}}}$) more randomness than would ever be needed in practice, so one might consider it effectively an ``infinite'' randomness expansion protocol. However, such a protocol avoids the need to reuse devices, and hence sidesteps the need for Input Security -- but secure device reuse is the key conceptual issue that we find interesting! 


Finally, the work \cite{csw14} serves as one very interesting example (discovered independently of this work) of how the concept of Input Security is relevant to problems other than infinite randomness expansion.  We note that our result can be combined with a quantum-secure randomness amplification protocol  (for example \cite{csw14}, or \cite{brandao2013robust}) to produce an infinite randomness amplification protocol.

\subsection{Barriers to infinite randomness expansion}
Here we identify the inherent technical challenges in analyzing any adaptive randomness expansion protocol.  In Section \ref{sec:results} we discuss how to overcome these challenges. Some of the technical issues discussed here have been identified in previous work (\eg,~\cite{Fehr2013}) and in randomness expansion folklore. 

\textbf{The Extractor Seed and Input Security Problems}

In any adaptive randomness expansion scheme there is a stage when intermediate outputs of the protocol are used to generate ``derived'' inputs for some devices in future stages of the protocol. This creates an inherent difficulty in analyzing adaptive protocols, because the devices involved in the protocol may adversarially take advantage of memory and shared entanglement to attempt to create harmful correlations between intermediate outputs and the the internal state of the devices that receive the ``derived'' inputs. To prove the correctness of an adaptive randomness expansion protocol, one must show that the devices receiving these ``derived'' inputs cannot distinguish them from inputs generated by a truly private random seed. Because of this fundamental challenge, there are very few analyses of adaptive randomness expansion protocols (or key distribution protocols for that matter) in the existing literature. Prior to our work,~\cite{Fehr2013} gave the only analysis of an adaptive randomness expansion protocol. However, their analysis requires the assumption that entanglement is only shared between certain pairs of devices, but otherwise that the devices are unentangled.

In the general case where devices can share arbitrary entanglement and may be entangled with an eavesdropper, we face the issue of the \emph{quantum security} of the intermediate outputs against devices that will receive the derived inputs\footnote{We say that a string $X$ is quantum secure, or simply secure, against an eavesdropper $E$ if the joint state of the string and eavesdropper $\rho_{XE}$ is approximately equal to $U_{|X|} \otimes \rho_E$, where $U_m$ denotes the uniform distribution on $|X|$ bits.}. This issue manifests itself in two different forms: the Input Security Problem and the Extractor Seed Problem. 

Generally, a randomness expansion protocol is comprised of two components: an expansion component and an extractor component. The expansion component will generate an output string that, while not necessarily close to uniformly random, will be guaranteed to have high min-entropy. The extractor component will then take this high min-entropy source, as well as a small polylogarithmic-sized uniformly random seed (taken, for example, from the initial seed of the randomness expansion protocol), and convert the high min-entropy source into a string that is close to uniform.

\textbf{The Input Security Problem}.  In an adaptive protocol, we require that the output of the expansion component contains high min-entropy \emph{relative to a quantum eavesdropper} (\ie high conditional min-entropy) -- where we treat the other devices in the protocol, collectively, as the eavesdropper. However, the Vazirani-Vidick protocol -- an quantum-secure exponential randomness expansion protocol that produces an output with high conditional min-entropy\footnote{Recent work by~\cite{ms14} gives another such protocol with quantum security. See Section~\ref{sec:related} for more information.} -- uses, in its analysis, an assumption that the initial seed to the protocol is secure against the eavesdropper~\cite{Vazirani2012}.  This is a condition that \emph{cannot} be satisfied in an adaptive protocol. Suppose in an adaptive protocol some device $D$ produced an intermediate output $X$, which we use as the derived input to some other device $D'$ as input randomness. Note that $X$ is \emph{not} secure against $D$. Hence, we cannot use the analysis of~\cite{Vazirani2012} as is and treat $D$ as an eavesdropper, and argue that $D'$ produces an output $Y$ that is secure against $D$. We refer to this issue as the Input Security Problem.

\textbf{The Extractor Seed Problem}. Even supposing that we had an expansion component that was immune to the Input Security Problem (\ie produces output that contains high conditional min-entropy despite the input being known to the eavesdropper), we would still suffer from a similar problem with the extractor component. Here, we need to use a small polylogarithmic-sized uniform extractor seed to convert a source of high conditional min-entropy into a string that is nearly uniform, relative to a quantum adversary. 

First, note that we cannot always take the extractor seed from the original random seed to the protocol, because this would limit us to exponential randomness expansion. Thus to achieve super-exponential expansion, the extractor seed must eventually be generated by intermediate outputs of the protocol. 

Secondly, the existing quantum-secure extractors in the literature (\eg, see~\cite{de2012trevisan,konig2008bounded,renner2008security}) require that the extractor seed be secure against the quantum eavesdropper. As pointed out by~\cite{Fehr2013}, provably satisfying this requirement in an adaptive randomness expansion protocol involves overcoming a technical difficulty similar to that of the Input Security Problem.  We refer to this technical barrier as the Extractor Seed Problem.

To summarize, in order to obtain quantum security of the output against an eavesdropper $E$, current quantum-secure expansion protocols and extraction procedures require the strong assumption that the joint state of the seed, the devices, and the eavesdropper $\rho_{SDE}$ is such that $\rho_{SDE} \approx U_{|S|} \otimes \rho_{DE}$, where $U_{|S|}$ denotes the uniform distribution on $|S|$ bits, and $\rho_{DE}$ denotes the internal state of the devices and adversary. In order to solve the Input Security and Extractor Seed Problems, we require randomness expansion protocols and extraction schemes that work with the weaker assumption that $\rho_{SD} \approx U_{|S|} \otimes \rho_D$ -- with no mention of the eavesdropper! -- while still obtaining the same quantum-security guarantees. We call this property \emph{Input Security}, and say that protocols with this property are \emph{Input Secure}.

It is interesting to note that extractors, by themselves, cannot satisfy a property like Input Security (i.e. we cannot guarantee that an extractor will produce private randomness when the seed is prepared by the adversary)\footnote{Here's a counter-example: let $D$ be an $n$-bit source that is uniformly random. Let $S$ be a $O(\log n)$-bit seed that is uniform and independent of $D$. Let $E$ denote the string $(S,\text{first bit of $\mathrm{Ext}(D,S)$})$. The min-entropy of $D$ with respect to $E$ is at least $n-1$, and $S$ is uniform and independent of $D$. However, the output of the extractor is \emph{not} secure against $E$.}.

The primary conceptual contribution of our paper is the design and analysis of the first randomness expansion protocols and extraction schemes that are (provably) Input Secure.
\medskip \\
\textbf{The Conditioning Security Problem}

The output guarantees of a randomness expansion protocol only hold \emph{conditioned} on the protocol succeeding (\ie conditioned on the event that the referee does not abort). Thus, the analysis of the security properties of the output of a protocol must take into account the fact that conditioning can skew the distribution of the output. Adversarially designed devices may, for example, coordinate to pass the protocol only when the first bit of the output is ``1''. This alone does not harm the min-entropy of the output by much, but suggests that there could be other strategies employed by adversarial devices to significantly weaken the security of the output. In~\cite{Vazirani2012}, they show that such a collusion strategy would imply that the eavesdropper and the devices could communicate with each other, a contradiction. However, this analysis again relies on the assumption that the initial seed is secure against the eavesdropper. When analyzing an Input Secure protocol, we cannot use this assumption, so resolving this Conditioning Security Problem requires different techniques. 
\medskip \\
\textbf{The Compounding Error Problem}

Another technical concern is the problem of error accumulation in an adaptive protocol. When using intermediate outputs to generate derived inputs for later stages in the protocol, we can only assume, at best, that the derived inputs are \emph{approximately} secure and uniform. Furthermore, these errors will accumulate over the course of the protocol, and in an infinite expansion protocol, this accumulation could grow so large that the protocol will fail to work at some point. Depending on how one measures the security of a string against an quantum eavesdropper, errors may not accumulate in a linear fashion --  as pointed out by~\cite{konig2007small}, even if the \emph{accessible information} of a string relative to an eavesdropper (which has been used as a standard security measure in quantum key distribution) is small, a tiny piece of classical side information could completely break the security of the string. Such an ill-behaved measure of quantum security would severely complicate the analysis of an adaptive randomness protocol. 

\section{Results} \label{sec:results}
We present a protocol that attains \emph{infinite randomness expansion}. Our protocol, which we denote the $\InfiniteExp$ protocol, involves a constant number of non-signaling devices (eight, specifically) that, with $m$ bits of seed randomness, can produce an arbitrarily large amount of certified randomness. In particular, starting with $m$ bits of random seed, if $\InfiniteExp$ is run for $k$ iterations, the output of the $k$ iterations is a random string that is $\exp(-\Omega(m^{1/3}))$-close to uniform, and has length 
$$
	\underbrace{2^{2^{\cdot^{\cdot^{2^{\Omega(m^{1/3})}}}}}}_{k},
$$
\ie, a $k$-height tower of exponentials in $m$. The initial seed length $m$ controls soundness parameters of the protocol, but \emph{has no bearing on the amount of certified output randomness!} 

Our protocol uses as subroutines the exponential expansion protocol of~\cite{Vazirani2012} (which we denote $\VV$)\footnote{We implicitly include the extraction procedure as part of the $\VV$ protocol, where the extractor seed is taken from the input seed of the $\VV$ protocol.}, and the sequential CHSH game protocol of Reichardt, \etal~\cite{Reichardt2012} (which we denote $\RUV$). See Section~\ref{sec:proto} for more detail on these sub-protocols. We describe the protocol below, both algorithmically and schematically (see Figure~\ref{fig:protocol}).

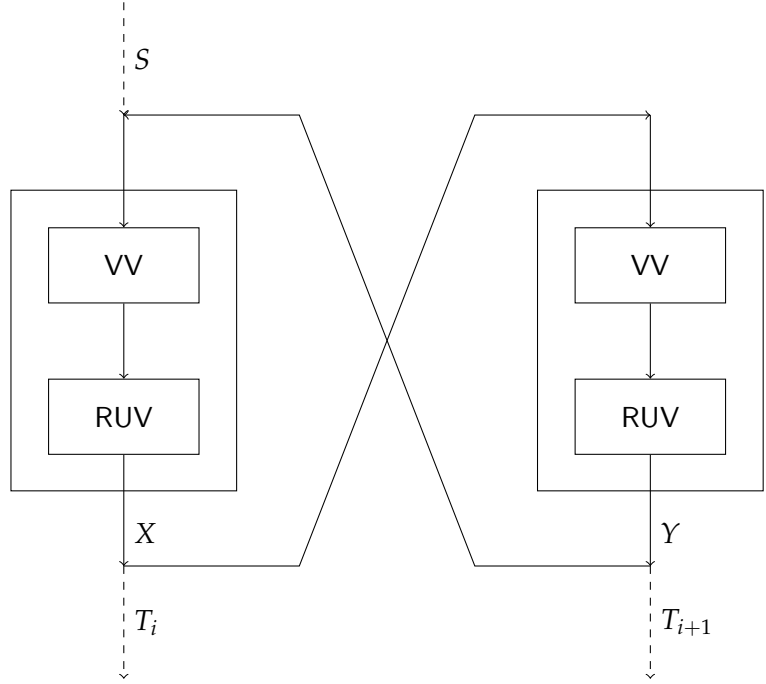
\begin{figure}[H] 
\begin{center}
\begin{tikzpicture}

\node[draw, minimum width=2cm,minimum height=1cm] (VV1) at (2,4.5) {$\VV$};
\node[draw, minimum width=2cm,minimum height=1cm,below = of VV1] (RUV1) {$\RUV$};
\draw (0.5, 1.5) rectangle  (3.5,5.5);
\draw [->] (VV1) -- (RUV1);

\node[draw, minimum width=2cm,minimum height=1cm] (VV2) at (9,4.5) {$\VV$};
\node[draw, minimum width=2cm,minimum height=1cm,below = of VV2] (RUV2) {$\RUV$};
\draw (7.5, 1.5) rectangle (10.5,5.5);
\draw [->] (VV2) -- (RUV2);

\path[draw,dashed,->] (2, 8) -- node[right](S) {$S$} (2, 6.5);
\path[draw,dashed,->] (2, 0.5) -- node[right](Ti) {$T_i$} (2, -1);
\path[draw,dashed,->] (9, 0.5) -- node[right](Ti2) {$T_{i+1}$} (9, -1);

\path[draw,->] (RUV1) -- node[right,pos=0.7] (X) {$X$} (2, 0.5);
\path[draw,->] (2, 0.5) -- (2 + 7/3.0, 0.5) -- (9 - 7/3.0, 6.5) -- (9, 6.5);
\path[draw,->] (9, 6.5) -- (VV2);

\path[draw,->] (RUV2) -- node[right,pos=0.7] (Y) {$Y$}  (9, 0.5);
\path[draw,->] (9, 0.5) -- (9 - 7/3.0, 0.5) -- (2 + 7/3.0, 6.5) -- (2, 6.5);
\path[draw,->] (2, 6.5) -- (VV1);

\end{tikzpicture}
\end{center}
\caption{The $\InfiniteExp$ protocol. All arrows indicate classical operations performed by the referee. $S$ denotes the initial seed to the protocol, and $T_i$ denotes the output of the protocol at the $i$th iteration. Each of the $\VV$ and $\RUV$ boxes involve two devices, for a total of eight devices used in the protocol.}
\label{fig:protocol}
\end{figure}

\begin{figure}[H]
\begin{center}
\framebox{
\begin{minipage}{0.9\textwidth}
\textbf{Non-signaling devices}: $D_1,\ldots,D_8$. \\
\textbf{Initial seed randomness}: $S \sim U_m$.

\begin{enumerate}
\item Let $X_1 \leftarrow S$.
\item For $i = 1, 2, 3, \ldots$
	\begin{enumerate}
		\item $Y_i \leftarrow \mathsf{VV}(D_1,D_2,X_i)$.
		\item $Z_i \leftarrow \mathsf{RUV}(D_3,D_4,Y_i)$.
		\item $W_i \leftarrow \mathsf{VV}(D_5,D_6,Z_i)$.
		\item $X_{i+1} \leftarrow \mathsf{RUV}(D_7,D_8,W_i)$.
	\end{enumerate}
\end{enumerate}
\end{minipage}
}
\end{center}
\caption{The algorithmic specification of the $\InfiniteExp$ protocol. $\VV(A,B,X)$ (resp. $\RUV(A,B,X)$) denotes executing the $\VV$ (resp. $\RUV$) sub-protocol with devices $A$ and $B$ using seed randomness $X$ (for more details about these sub-protocols see Section~\ref{sec:proto}). The $X_i$, $Y_i$, $Z_i$, and $W_i$ registers are all classical, and managed by the referee.}
\end{figure}

The main result of this paper is the following theorem, stated informally here (for the formal version see Theorems~\ref{thm:infinite-exp} and~\ref{thm:infinite-exp-completeness}):

\begin{theorem}[Infinite randomness expansion, informal]
\label{thm:infinite-exp-informal}
	Let $D = \{D_1,\ldots,D_8\}$ denote eight non-signaling quantum devices. Let $E$ be an arbitrary quantum system that may be entangled with the $D_i$'s, but cannot communicate with them. Suppose that a classical referee executes the $\InfiniteExp$ protocol with the $\{D_i\}$ devices, using an $m$-bit random seed $S$ that is secure against the devices $\{D_i\}$. Then, for all $k \in \N$, if $\Pr(\text{Protocol has not aborted by round $k$}) = \exp(-O(m^{1/3}))$, then the output $T_k$ of the protocol, conditioned on not aborting after $k$ rounds, is $\exp(-\Omega(m^{1/3}))$-secure against $E$, and has length $\Omega(g^{(k)} (m))$, where $g^{(k)}$ denotes the $k$-fold composition of the function $g:\N \to \N$, defined as $g(m) = \exp(\Omega(m^{1/3}))$. 
	
Furthermore, there exists a quantum strategy for the devices such that, with high probability, they do not abort the protocol at any round.
\end{theorem}

The analysis of the $\InfiniteExp$ protocol overcomes the challenges described in the previous section. We now give an overview of how we solve them.

\subsection{Our proof strategy}
\textbf{Solving the Extractor Seed and Input Security Problems}. The key technique for solving both the Extractor Seed and Input Security Problems is a powerful result of Reichardt, Unger, and Vazirani~\cite{Reichardt2012}, which is based on the phenomenon of \emph{CHSH game rigidity}. The CHSH game is a two-player game in which a classical referee chooses two input bits $x$ and $y$ uniformly at random, and gives them to non-communicating players Alice and Bob. Alice and Bob produce binary outputs $a$ and $b$, and they win the game if $a \oplus b = x \wedge y$. If Alice and Bob employ classical strategies, they cannot win the CHSH game with probability exceeding $75\%$, but using shared quantum entanglement, there is a quantum strategy that allows them to win the game with probability $\cos^2 (\pi/8) \approx 85\%$. The CHSH game is frequently used in the study of quantum entanglement and non-locality. More relevantly, it also serves as the basis for many randomness expansion protocols in the literature: protocols will often test for Bell inequality violations by measuring how often devices win the CHSH game.

The famous Tsirelson's Theorem states that $\cos^2(\pi/8)$ is the optimal winning probability using quantum strategies.  Even more remarkable is that the CHSH game is \emph{rigid}: there is essentially a \emph{unique} quantum strategy that achieves this optimum. That is, any quantum strategy that achieves $\cos^2(\pi/8)$ winning probability must be, in a specific sense, isomorphic to the ``canonical'' CHSH strategy which involves Alice and Bob making specific measurements on separate halves of an EPR pair\footnote{The EPR pair state is defined as $\ket{\psi} = \frac{1}{\sqrt{2}} \left( \ket{00} + \ket{11} \right)$.} (which we will call the \emph{ideal CHSH strategy}). Furthermore, CHSH game rigidity is robust: any strategy that achieves $\cos^2(\pi/8) - \eps$ winning probability must be isomorphic to a strategy that is $O(\sqrt{\eps})$-close to the ideal CHSH strategy. A form of CHSH game rigidity was first proved by Mayers and Yao in the exact case~\cite{mayers2003self} and later made robust by~\cite{mys:2012,MillerS:self-testing:2013}.

Reichardt \etal proved a far-reaching generalization of CHSH game rigidity to the situation where Alice and Bob play $N$ independent CHSH games in sequence. This can be viewed as a larger game $\CHSHN$, where Alice and Bob win $\CHSHN$ if they win approximately $\cos^2(\pi/8)N$ games. Reichardt \etal prove the following theorem, stated informally here (for the precise version see \cite{Reichardt2012} Theorem 5.38, or Theorem 2.8 in this paper), which they call \emph{sequential CHSH game rigidity}:

\begin{theorem}[Sequential CHSH game rigidity, informal version]
\label{thm:ruv-informal}
Suppose Alice and Bob play $N$ instances of the CHSH game, where the inputs to Alice and Bob in each instance are uniform and independent of each other. Divide the $N$ instances into $N/t$ blocks of $t$ games each, where $t = N^{1/\alpha}$ for some universal constant $\alpha > 1$. If Alice and Bob use a strategy that, with high probability, wins approximately $\cos^2(\pi/8)N$ instances, then in most blocks, Alice and Bob's strategy is approximately isomorphic to the ideal sequential strategy, in which the ideal CHSH strategy is applied $t$ times in sequence to $t$ EPR pairs that are in tensor product with each other.
\end{theorem}

Sequential CHSH game rigidity is a powerful tool that allows one to characterize the behavior of separated quantum devices, simply from observing the correlations between their (classical) inputs and outputs. Reichardt \etal use sequential CHSH games as a primitive in a more general protocol that allows a classical computer to command non-signaling quantum devices to perform arbitrary quantum computation -- and verify that this computation has been performed correctly! Here, in contrast, our goal is much more modest: we simply want to command non-signaling quantum devices to generate uniformly random bits.  

The $\CHSHN$ game already yields a protocol that produces certified randomness. In particular, we have two non-signaling devices play $N$ games of CHSH. The referee will check whether the devices won approximately $\cos^2(\pi/8)N$ games. If so, the referee will select a block of $t$ games at random, and use the output of one of the devices in that block of $t$ games be the protocol's output -- call this the $\RUV$ protocol. 

We know from Theorem~\ref{thm:ruv-informal} that, with high probability, the outputs of the $\RUV$ protocol were generated by a strategy approximating the ideal sequential strategy. The ideal sequential strategy is the ideal CHSH measurement repeatedly applied to a tensor product of EPR pairs, so the measurement outcomes are necessarily in tensor product with an eavesdropper. Thus the outputs of $\RUV$ are approximately secure against a quantum adversary. The problem, of course, is that the amount of randomness needed by the referee to run this $\RUV$ protocol is much greater than the amount of certified randomness in the output ($\Theta(N)$ versus $N^{1/\alpha}$).  So we can't use $\RUV$ by itself as a randomness expansion scheme.  

However, sequential CHSH game rigidity offers more than just the guarantee of secure uniform randomness; observe that it \emph{does not need to assume that the inputs to the $N$ CHSH games were secure against an eavesdropper} -- only that it was secure against the devices playing the CHSH games! This is precisely the Input Security property. 

Thus, we can use the $\RUV$ protocol as a ``scrambling'' procedure that transforms an input that may not be secure against an eavesdropper into a shorter string that \emph{is} secure against an eavesdropper. Recall that, because of the Input Security and Extractor Seed Problems, the output of the $\VV$ sub-protocol in the $\InfiniteExp$ protocol may not be secure against other devices (namely, the devices that produced the input to the $\VV$ sub-protocol). However, if we invoke the $\RUV$ protocol on the outputs of $\VV$, we obtain secure outputs that can be used as input randomness for another $\VV$ instance. 

Furthermore, observe that we still have achieved randomness expansion: the $\VV$ protocol attains exponential expansion, and the $\RUV$ protocol will only shrink that by a polynomial amount.

\textbf{Solving the Conditioning Security Problem}. The main technical contribution of our paper is solving the Conditioning Security Problem. While combining the $\VV$ and $\RUV$ protocols conceptually yields an Input Secure randomness expansion protocol, there still is the technical issue of whether this protocol is Input Secure when we condition on the $\RUV$ protocol succeeding. There are simple examples that show that adversarial devices can, via conditioning, skew the distribution of their outputs, and even introduce entanglement between some bits of their outputs and an eavesdropper, despite most outputs having been produced by an ideal strategy. The Sequential CHSH Game Rigidity Theorem of~\cite{Reichardt2012} does not take conditioning into account, because it is assumed that the devices pass the $\RUV$ protocol with probability extremely close to $1$. 

Here, we assume the $\RUV$ protocol passes with some small probability that is inverse polynomial in the number of games played, and show that the $\RUV$ protocol manages to obtain an approximately secure output conditioned on the protocol succeeding. We prove this in Lemma~\ref{lem:ruv-scrambler}, and our proof employs tools from quantum information theory. Our approach is reminiscent of that used in the proofs of the classical Parallel Repetition Theorem (see, \eg,~\cite{holenstein2007parallel}). 

\textbf{Solving the compounding error problem}. We use the strongest definition of the quantum security of a string against an eavesdropper: namely, a string $X$ is (approximately) secure against an eavesdropper $E$ iff the trace distance between the joint state $\rho_{XE}$ and the ideal state $U_{|X|} \otimes \rho_E$ is small, where where $U_{|X|}$ denotes the uniform distribution on $|X|$ bits. To solve the compounding error problem, we first show that the errors incurred at each iteration of the $\InfiniteExp$ protocol accumulate \emph{linearly} -- this is because the trace distance satisfies the triangle inequality. Then, we show that the error added at iteration $k$ is \emph{exponentially} smaller than the error of iteration $k-1$. Thus, the infinite sum of errors converges to a constant multiple of the error incurred by the first iteration, which is exponentially small in the seed length $m$. Hence we avoid the potential problems raised by~\cite{konig2007small}. 

\subsection{Related work}\label{sec:related}

Here we discuss some relevant recent developments in the area of randomness expansion and amplification, which were announced after the original posting of this work.  We note, however, that the results in the following works were discovered independently of the results in this work, and their relationship to each other was only realized after both works were essentially complete.  In the following description we will occasionally use the terminology of this paper to restate results of these other works, though those papers used different terminology in the original statements.

In independent work by Chung, Shi, and Wu \cite{csw14}, the problem of Input Security was also studied, and played a key role in their construction of a device-independent protocol to amplify randomness, starting with any min-entropy source.  The authors require an Input Secure randomness expansion protocol to use as a building block for their amplification protocol. They prove an elegant result called the Equivalence Lemma, which may be informally summarized as follows  (see \cite{csw14} for a formal statement): 

     Consider a device-independent randomness expansion protocol $P$, that starts with a seed $S$, uniform and in tensor product with the devices $D$ involved in the protocol, as well as a quantum adversary $E$, and produces an output string $X$ that is certifiably close to uniform and in tensor product with $E$ and $S$.  The Equivalence Lemma states that any such protocol $P$ \emph{also} certifies output randomness $X$ with the same security guarantees, \emph{without requiring that $S$ is in tensor product with $E$} — in other words, any such protocol $P$ is also Input Secure. In particular, this proves that the Vazirani-Vidick protocol (when implemented in composition with a strong quantum extractor) is, in fact, Input Secure, and can be composed with itself to perform unbounded randomness expansion in the same manner as we do here, without requiring the use of the RUV protocol. 

Secondly, another independent work of Miller and Shi \cite{ms14} gives the first provably robust protocol for randomness expansion (and, in fact, gives robust exponential expansion).  Combining the main result of \cite{ms14} with Equivalence Lemma of \cite{csw14}, allows one to obtain a provably  \emph{robust} infinite expansion protocol requiring only four non-communicating devices.

It is interesting to note that extractors (which have a similar input-output structure to randomness expansion protocols) cannot possess an analogous Input Security. Thus, there is no natural analogue of the Equivalence Lemma which will work for extractors.  In this sense, the Equivalence Lemma represents an interesting phenomenon or property which is possessed by device independent (quantum) protocols, but not by (classical) protocols such as extractors.

\section{Preliminaries}\label{sec:prelim}
\subsection{Notation}

We write $[N]$ for the set of integers $\{1,\ldots,N\}$. For a Hilbert space $\Hilb$, let $D(\Hilb)$ denote the set of density matrices on $\Hilb$. The classical state $\rho_X$ corresponding to a discrete classical random variable $X$ is defined as $\sum_x \Pr(X = x) \ketbra{x}{x}$ (where $x$ ranges over the computational basis states). For a discrete classical random variable $X$, we use $|X|$ to denote $X$'s length in bits. A classical-quantum state (or \emph{cq-state}) $\rho_{XB}\in D(\Hilb_X \otimes \Hilb_B)$ is a density matrix where $\rho_{XB} = \sum_{x} p_x \ketbra{x}{x} \otimes \rho_B^{x}$, where $p_x$ are probabilities and $\{ \ket{x} \}$ is an orthonormal basis for $\Hilb_X$. We write $\ident_N$ to denote the $N \times N$ identity matrix. We write $U_m$ to denote the density matrix $2^{-m} \ident_{2^m}$ (i.e. the completely mixed state of dimension $2^m$). For an arbitrary matrix $A$, we let $\|A \|_{\tr} := \frac{1}{2} \tr \sqrt{A^{\dagger} A}$ denote its trace norm (also known as its Schatten $1$-norm). 

\begin{definition}[Secure cq-state]
\label{def:ideal-block}
	Let $E$ be an arbitrary quantum system. Let $\rho_{XE}$ be a cq-state.  For state $\rho_{XE}$, $X$ is \textbf{$\zeta$-secure against $E$} iff 
	$$
		\| \rho_{X E} - U_{|X|} \otimes \rho_E \|_{\tr} \leq \zeta.
	$$
\end{definition}

\subsection{Quantum information theory} 

For completeness we present a few key definitions and facts of quantum information theory that will be useful for us later. For a more comprehensive reference we refer the reader to, \eg,~\cite{nielsen2010quantum, wilde2013quantum}. 

For a density matrix $\rho$, its von Neumann entropy is defined as $H(\rho) := -\tr(\rho \log \rho)$. For a density matrix $\rho_{AB} \in D(\Hilb_A \otimes \Hilb_B)$, the conditional von Neumann entropy is defined as $H(A  |  B)_\rho := H(AB)_\rho - H(B)_\rho$ where $H(AB)_\rho = H(\rho_{AB})$ and $H(B)_\rho = H(\rho_B)$. The quantum mutual information between $A$ and $B$ of $\rho_{AB}$ is defined as $I(A : B)_\rho := H(A)_\rho - H(A  |  B)_\rho$. The conditional quantum mutual information $I(A : B  |  C)_\rho$ for a tripartite state $\rho_{ABC}$ is defined as $H(A  |  C)_\rho - H(A  |  B, C)_\rho$. We will usually omit the subscript $\rho$ when the state is clear from context.

We now list a few useful facts about these quantum information-theoretic quantities. Proofs of the following facts can be found in, \eg,~\cite{wilde2013quantum}.
\begin{fact}
\label{fact:qinfo}
\begin{enumerate}
	\item Let $X$ be a discrete random variable, and let $\rho_X$ be its associated classical state. Then $H(\rho_X) = H(X)$, where $H(X)$ is the Shannon entropy of $X$.
	\item (Conditioning reduces entropy) Let $\rho_{AB} \in D(\Hilb_A \otimes \Hilb_B)$. Then $H(A|B)_\rho \leq H(A)_\rho$.
	\item (Chain rule) Let $\rho_{ABC} \in D(\Hilb_A \otimes \Hilb_B \otimes \Hilb_C)$. Then 
	$$
		I(A:BC)_\rho = I(A:B)_\rho + I(A:C | B)_\rho.
	$$
	\item (Pinsker's inequality) Let $\rho_{AB} \in D(\Hilb_A \otimes \Hilb_B)$. Then
	$$
		\| \rho_{AB} - \rho_A \otimes \rho_B \|_{\tr}^2 \leq 2 I(A : B)_\rho.
	$$
\end{enumerate}
\end{fact}
Finally, we define quantum min-entropy. Let $\rho_{AB}$ be a bipartite density matrix. The min-entropy of $A$ conditioned on $B$ is defined as
$$
	\MinEntropy(A | B)_\rho := \max \{ \lambda \in\R: \exists \sigma_B \in D(\Hilb_B)\text{ s.t. } 2^{-\lambda} \ident_A \otimes \sigma_B \geq \rho_{AB} \}.
$$
Let $\eps > 0$. Then $\eps$-smoothed min-entropy of $A$ conditioned on $B$ is defined as
$$
	\MinEntropy^{\eps}(A | B)_\rho := \max_{\tilde{\rho}_{AB} \in B(\rho_{AB},\eps)} \MinEntropy(A | B)_{\tilde{\rho}},
$$
where $B(\rho_{AB},\eps)$ is the set of sub-normalized density matrices within trace distance $\eps$ of $\rho_{AB}$. For a detailed reference on quantum min-entropy, we refer the reader to~\cite{renner2008security}.

\subsection{Modelling protocols and input robustness}

In this paper, we will consider several different randomness expansion procedures (\eg, the Vazirani-Vidick protocol, or the $\RUV$ protocol); a crucial element of our analysis is that these protocols are all \emph{input robust} in the sense that slight deviations from uniformity in their input seed only mildly affect the expansion guarantees that we get when assuming the seed is perfectly uniform. To make this input robustness property formal, we introduce the quantum operation description of randomness expansion protocols. 

In general, a randomness expansion protocol is an interaction between a classical referee $R$ and a quantum device $D$, that is entirely unconstrained, except that $D$ consists of two or more isolated, non-signaling sub-devices (but the sub-devices may be entangled).

The important Hilbert spaces we will consider are:
\begin{enumerate}
	\item (\textbf{Pass/No Pass Flag}). $\Hilb_F$ denotes a two-dimensional Hilbert space that the referee will use to indicate whether it accepts or rejects the interaction.
	\item (\textbf{Protocol seed}). $\Hilb_S$ denotes the $2^m$-dimensional Hilbert space that corresponds to the (private) $m$-bit seed randomness that the referee will use for its interaction with the device $D$.
	\item (\textbf{Protocol output}). $\Hilb_X$ denotes the Hilbert space that corresponds to the output of the device $D$\, \footnote{Since $D$ always consists of non-signaling subdevices, we will arbitrarily declare one of the sub-devices' output to be the output of the overall device $D$.}.
	\item (\textbf{Device internal state}). $\Hilb_D$ denotes the Hilbert space corresponding to the internal state of the device $D$. 
	\item (\textbf{Eavesdropper}). $\Hilb_E$ denotes the Hilbert space corresponding to a potential quantum eavesdropper, which may be entangled with device $D$.
\end{enumerate}

We can view a randomness expansion protocol as a quantum operation $\mathcal{E}$ acting on states in the space $\Hilb_F \otimes \Hilb_S \otimes \Hilb_X \otimes \Hilb_D$. Of the Hilbert spaces listed above, device $D$ only has access to the Hilbert space $\Hilb_D$; the other Hilbert spaces get updated by the referee's interaction with $D$ (except for $\Hilb_E$ which is controlled by the eavesdropper). For example, the referee, by interacting with $D$, will write $D$'s outputs to register $X$. The states in the Hilbert spaces $\Hilb_F$, $\Hilb_S$, and $\Hilb_X$ will always be classical mixed states (\ie diagonal in the computational basis).

More precisely, let $P$ be a randomness expansion protocol. We will model $P$ as a quantum operation $\mathcal{E}$ acting on an initial state $\rho^i_{FSXD}$ in the space $\Hilb_F \otimes \Hilb_S \otimes \Hilb_X \otimes \Hilb_D$, where $\rho^i_D$ is the internal state of $D$ before the protocol starts, and $\rho^i_{FSX}$ is prepared by the referee. $\mathcal{E}$ will be some unitary map $V_P$ applied to the joint state $\rho^i_{FSXD}$. Now, define the quantum operation $\mathcal{F}$ that takes a state $\rho_{FSXD}$, and produces the post-measurement state of $\rho_{FSXD}$ \emph{conditioned} on measuring $\ket{1}$ in the $F$ register, and then traces out the $F$ and $S$ registers, leaving $\rho_{XD|F=1}$. We define $\mathcal{FE}$ to be the composition of the two quantum operations $\mathcal{E}$, followed by $\mathcal{F}$.  Throughout this paper, we will decorate density matrices by superscripts $i$ and $f$ to denote the states before and after the protocol, respectively. For example, we will often let $\rho^f_{FSXD}$ denote the state of the $FSXD$ system after the execution of the protocol, conditioned on the protocol succeeding (\ie $F=1$).

The completeness and soundness of protocol $P$ are statements about the post-measurement state $\mathcal{FE} \otimes \ident_E (\rho^i_{FSXDE})$ (where $\ident_E$ is the identity on $\Hilb_E$), argued only with respect to an \emph{ideal} initial state $\rho^i_{FSXDE}$ such that $\rho^i_{FSXD} := \ketbra{0}{0}_F \otimes U_m \otimes \ketbra{0}{0}_X \otimes \rho^i_{D}$, (or, depending on the analysis, the stronger assumption that $\rho^i_{FSXDE} := \ketbra{0}{0}_F \otimes U_m \otimes \ketbra{0}{0}_X \otimes \rho^i_{DE}$). In other words, the initial seed is assumed to be perfectly uniform and unentangled with the device $D$. However, we also have a form of input robustness: if the initial state were instead $\delta$-close in trace distance to the ideal initial state defined above, then we would obtain the same output parameters as $P$, up to an $\delta/\lambda$ additive factor in trace distance, where $\lambda$ is the probability that $\ket{1}$ is measured in the $F$ register. We prove this formally in Lemma~\ref{lem:approx} below.

\begin{lemma}
\label{lem:approx}
Let $D$ be a device, and $E$ an arbitrary quantum system that may be entangled with $D$. Let $\sigma_{FSX} := \ketbra{0}{0}_F \otimes U_{|S|} \otimes \ketbra{0}{0}_X$. Let the quantum operations $\mathcal{F}$, $\mathcal{E}$, and $\mathcal{FE}$ be defined as above. Suppose for all states $\sigma_{FSXDE}$ such that $\sigma_{FSXD} = \sigma_{FSX} \otimes \sigma_{D}$, there exists a state $\tau_{XDE}$ such that $\tau_{XE} = U_{|X|} \otimes \sigma_E$ and 
$$
	\| \mathcal{FE} \otimes \ident_E (\sigma_{FSXDE}) - \tau_{XDE} \|_{\tr} \leq \eps.
$$

Let $\delta, \lambda > 0$. Let $\rho^i_{FSXDE}$ be such that $\| \rho^i_{FSXDE} - \sigma_{FSXDE} \|_\tr \leq \delta$ for a state $\sigma_{FSXDE}$ where $\sigma_{FSXD} = \ketbra{0}{0}_F \otimes U_{|S|} \otimes \ketbra{0}{0}_X \otimes \sigma_D$. Suppose that the probability of measuring $\ket{1}$ in the $F$ register for the state $\mathcal{E} \otimes \ident_E(\rho^i_{FSXDE})$ is at least $\lambda$. Then, there exists a state $\mu_{XDE}$ such that $\mu_{XE} = U_{|X|} \otimes \mu_E$ and
$$
	\| \rho^f_{XDE} - \mu_{XDE} \|_\tr \leq \eps + \delta/\lambda,
$$
where $\rho^f_{XDE} := \mathcal{FE}\otimes \ident_E (\rho^i_{FSXDE})$.
\end{lemma}

The proof of Lemma~\ref{lem:approx} is deferred to Appendix~\ref{app:cond}.

\subsection{The Vazirani-Vidick protocol and quantum-secure extractors}

Vazirani and Vidick exhibit a protocol that involves two non-signaling quantum devices and a classical referee, that achieves randomness expansion that is secure against a quantum eavesdropper~\cite[Protocol B]{Vazirani2012}. We record a formulation of their result as it will be used by us here:

\begin{theorem}[Vazirani-Vidick protocol~\cite{Vazirani2012}]
\label{thm:vv}
There exists a protocol $P$ with the following properties. Let $D_1$ and $D_2$ be arbitrary non-signaling quantum devices. Let $E$ be an arbitrary quantum system, possibly entangled with $D_1$ and $D_2$, but cannot communicate with $D_1$ and $D_2$ once the protocol begins. The protocol, executed with devices $D_1$ and $D_2$, has the following properties:

\begin{enumerate}
	\item (Output length). The output of the protocol has length $n(m) = \exp(Cm^{1/3})$, for some constant $C$;
	\item (Completeness). There exists a non-signaling quantum strategy for $D_1$ and $D_2$ to pass the protocol with probability $1 - \exp(-\Omega(m^{2/3}))$;
	\item (Soundness). If the initial joint state $\rho^i_{SD_1D_2E}$ of the seed $S$, devices $D_1, D_2$, and eavesdropper $E$ is such that $\rho^i_{SD_1D_2E} = U_m \otimes \rho^i_{D_1D_2E}$, then if $\Pr(\text{Protocol succeeds}) \geq \eps$, we have that
	$$H^{\eps}_{\infty}(X  |  E)_{\rho^f} \geq h(m),$$
	 where $\eps = \eps(m)$, and $\rho^f_{XE} $ denotes the joint state of device $D_1$'s output and $E$, \emph{conditioned} on the protocol succeeding.
\end{enumerate}
where $h(m) := \exp(C' m^{1/3})$ and $\eps(m) := 1/h(m)$, for a universal constant $C'$. 
\end{theorem}

Another important primitive we will use is a \emph{quantum-secure extractor}. 

\begin{definition}[Quantum-secure extractor] A function $\mathrm{Ext}: \{0,1\}^n \times \{0,1\}^d \to \{0,1\}^r$ is a $(h,\eps)$-quantum-secure extractor iff for all cq-states $\rho_{XE}$ classical on $n$-bit strings $X$ with $H_{\infty}(X  |  E)_\rho \geq h$, and for uniform seed $S$ secure against $X$ and $E$ (that is, the joint state $\rho_{XES}$ is such that $\rho_{XES} = \rho_{XE} \otimes U_d$), we have
$$
	\big \| \rho_{\mathrm{Ext}(X,S)ES} - U_r \otimes \rho_{ES} \big \|_{\mathrm{tr}} \leq \eps,
$$
where $\rho_{\mathrm{Ext}(X,S)ES}$ denotes the joint cqc-state on the extractor output, quantum side information $E$, and the seed $S$.
\end{definition}

\begin{theorem}[\cite{de2012trevisan}] \label{thm::quantumextractor}
\label{thm:dpvr}
For all positive integers $n$, $r$, there exists a function $\mathsf{QExt}:\{0,1\}^n \times \{0,1\}^d \to \{0,1\}^r$ that is a $(r + O(\log r) + O(\log 1/\eps), \eps)$-quantum-secure extractor where $d = O(\log^2(n/\eps) \log r)$.
\end{theorem}

\subsection{Sequential CHSH game rigidity}

We can view a sequence of $N$ CHSH games, played by non-signaling quantum devices $D_1,D_2$, as a protocol $\CHSHN$, where the referee uses a private random seed $S$ to generate inputs $A_i, B_i\in \{0,1\}$ to the devices $D_1$ and $D_2$, and obtains their respective outputs $X_i,Y_i\in\{0,1\}$ for each round $i \in [N]$. The protocol succeeds if $W$, the number of rounds $i$ such that $X_i \oplus Y_i = A_i \wedge B_i$, is at least $(\cos^2(\pi/8) - O(\frac{\log N}{\sqrt{N}}))N$.

Divide the $N$ rounds of the $\CHSHN$ protocol into \emph{blocks} of $t$ consecutive games each, where $t = \lfloor N^{1/\alpha} \rfloor$ for some fixed constant $\alpha$. Let $X$ be the output register of device $D_1$.  Let $X_i$ denote the $t$-qubit register of the $i$th block of $X$.

We paraphrase the sequential CHSH game rigidity theorem of~\cite{Reichardt2012} here. In the theorem, we imagine that for each block of games, the devices $D_1$, $D_2$ apply some local quantum operation on their respective systems to produce outputs for the block. We call the quantum operation applied in each block $i$ their \emph{block strategy} for $i$. We say that a block strategy is $\zeta$-ideal if there is a local isometry $\mathcal{I}$ under which their quantum operation $\mathcal{E}$ and the state acted upon by $\mathcal{E}$ are together $\zeta$-close to the ideal CHSH strategy (for a precise definition of $\zeta$-ideal strategies, see~\cite{Reichardt2012}). The main property of $\zeta$-ideal strategy that we will use is the following:

\begin{lemma}
\label{lem:ideal_strategy}
	Let $D_1,D_2$ be non-signaling quantum devices. Suppose that $D_1$ and $D_2$ participate in the $\CHSH{N}$ protocol. Let $E$ be an arbitrary quantum system that may be entangled with $D_1$, $D_2$, but cannot communicate with them once the $\CHSHN$ protocol begins. Let $I_i$ be the indicator random variable denoting whether $D_1$ and $D_2$'s block strategy for block $i$ is $\zeta$-ideal. Let $X_i$ be the output of block $i$. Then, 
	$$
		\| \rho_{X_i E | I_i = 1} - U_{n} \otimes \rho_{E | I_i = 1} \|_\tr \leq \zeta,
	$$
	where $\rho_{X_i E | I_i=1}$ denotes the joint state of $X_i$ and $E$, conditioned on the event $I_i=1$.
\end{lemma}
\begin{proof}
	This is straightforward given the definition of $\zeta$-ideal strategy. See~\cite[Definitions 5.4, 5.5  and 5.37]{Reichardt2012} for more detail.
\end{proof}

\begin{theorem}[Sequential CHSH game rigidity; Theorem 5.38 of~\cite{Reichardt2012}]
\label{thm:ruv_sequential_rigidity}
	Let $D_1,D_2$ be non-signaling quantum devices. Suppose that $D_1$ and $D_2$ participate in the $\CHSH{N}$ protocol. Let $E$ be an arbitrary quantum system that may be entangled with $D_1$, $D_2$, but cannot communicate with them once the $\CHSHN$ protocol begins. Let $W$ be the total number of CHSH games that $D_1$ and $D_2$ win in the protocol. Let $X$ be the output of $D_1$. Fix $\eps > 0$, and let $G \leq N/t$ be the total number of blocks $i$ such that the strategy employed by $D_1$ and $D_2$ in block $i$ is $\kappa_{*} t^{-\kappa_*}$-ideal, where $\kappa_* >1$ is a universal constant. Then, 
	$$
		\Pr(W \geq \cos^2 (\pi/8) N - \frac{1}{2\sqrt{2}} \sqrt{ N \log N }\text{ and }G \leq (1 - \nu) N/t) \leq \frac{1}{t^2},
	$$
	where $\nu = (12/\sqrt{2}) \sqrt{\log N}t/N^{1/4} $, and $t > 85$.
\end{theorem}
\begin{proof}
	This is Theorem 5.38 of~\cite{Reichardt2012}, instantiated with the parameter settings used in Theorem 5.39.
\end{proof}

\section{The Protocol} \label{sec:proto}
In this section we formally define the protocol for infinite certifiable randomness expansion, which we call the $\InfiniteExp$ protocol. The protocol uses eight non-signaling devices, which may all share entanglement, but cannot communicate with each other. The devices are partitioned into two \emph{Expansion Clusters} $C_0$ and $C_1$ with four devices each. In each iteration, the $\InfiniteExp$ protocol alternates between clusters $C_0$ and $C_1$, performing a sub-protocol called $\ClusterExp$. The output of one cluster is used as seed randomness for the next invocation of the $\ClusterExp$ sub-protocol with the other cluster. Only the first iteration requires some seed randomness, to ``jumpstart'' the randomness expansion process.
 
\begin{figure}[H]
\begin{center}
\textbf{$\InfiniteExp$ Protocol} \\
\medskip
\framebox{
\begin{minipage}{0.9\textwidth}
\textbf{Non-signaling Clusters}: $C_0$, $C_1$. \\
\textbf{Initial seed randomness}: $S \sim U_m$.

\begin{enumerate}
\item Let $X_1 \leftarrow S$.
\item For $i = 1, 2, 3, \ldots$
	\begin{enumerate}
		\item $X_{i+1} \leftarrow \ClusterExp(C_{i},X_i)$.
                \item If $\ClusterExp$ aborts, then abort the entire protocol, otherwise continue.
	\end{enumerate}
\end{enumerate}
\end{minipage}
}
\end{center}
\caption[font=\tiny]{The $\InfiniteExp$ protocol. The classical registers $X_i$ are maintained by the referee, and $C_i$ denotes cluster $C_{i \mod 2}$. $X_{i+1} \leftarrow \ClusterExp(C_{i},X_i)$ denotes executing the $\ClusterExp$ sub-protocol with the devices in cluster $C_i$, using $X_i$ as the seed randomness, and storing the sub-protocol output in register $X_{i+1}$.}
\end{figure}

We now specify the sub-protocol $\ClusterExp(C,S)$ for a $4$-device cluster $C$ and seed randomness $S$. As discussed earlier, two devices of a cluster $C$ will be used to perform the Vazirani-Vidick near-exponential randomness expansion protocol, and the other two will be used to perform a variant of the $\CHSHN$ protocol, which we call the $\RUV$ protocol. 

\begin{figure}[H]
\begin{center}
\textbf{$\ClusterExp( C, S )$ Sub-Protocol} \\
\medskip
\framebox{
\begin{minipage}{0.9\textwidth}
\textbf{Input Non-signaling Devices}: $C := \{D_1,D_2,E_1,E_2\}$. \\
\textbf{Input seed randomness}: $S$

\begin{enumerate}
                \item $Y \leftarrow \VV(D_1,D_2,S)$.
		\item $Z \leftarrow \RUV(E_1,E_2,Y)$.
                \item If either of the above instances of $\mathsf{VV}$ or $\mathsf{RUV}$ aborts, then abort $\ClusterExp$.  Otherwise continue.
                \item Output $Z$.
\end{enumerate}
\end{minipage}
}
\end{center}
\end{figure}

It is important that no subset of these devices can communicate with (signal to) any other subset of the devices throughout the course of the subroutine.  We now give precise definitions of the $\VV$ and $\RUV$ sub-protocols.

\subsection{The $\VV$ sub-protocol}

The $\VV$ sub-protocol consists of performing Protocol B from~\cite{Vazirani2012}, and then applying a randomness extractor to the output of Protocol $B$. For any $s$, Protocol B takes in a uniformly random $s$-bit seed, and conditioned on the protocol succeeding, produces a string of length $n(s) =\exp(\Omega(s^{1/3}))$ with $h(s) = \exp(\Omega(s^{1/3}))$ bits of (smoothed) min-entropy (see Theorem~\ref{thm:vv}). We give a detailed account of the particular parameter settings we use for Protocol B in Appendix~\ref{app:params}. 

We use the $\QExt$ randomness extractor given by Theorem~\ref{thm:dpvr}. More formally, by $\QExt_{n,r,\eps}$ we denote the $(r + O(\log r) + O(\log 1/\eps),\eps)$-quantum-secure extractor mapping $\{0,1\}^n \times \{0,1\}^d$ to $\{0,1\}^r$, where $d = d(n,r,\eps) = O(\log^2(n/\eps) \log r)$. 

For all $s$, the $\VV$ sub-protocol takes in a $s$-bit seed $S$, and outputs $v(s)$ bits, where $v(s) := \exp(\Omega(s^{1/3}))$ (for more detail, see Appendix~\ref{app:params}).

\begin{figure}[H]
\begin{center}
\textbf{$\mathsf{VV}(A, B ,S)$ Sub-Protocol} \\
\medskip
\framebox{
\begin{minipage}{0.9\textwidth}
\textbf{Input Non-signaling Devices}: $A,B $

\textbf{Input Seed }: $S$ 

\begin{enumerate}
                \item Let $S_1$ be the first $\left \lfloor s/2 \right \rfloor$ bits of $S$, and $S_2$ be the last $\left \lfloor s/2 \right \rfloor$ bits of $S$, where $s := |S|$.
		\item Perform Protocol B of~\cite{Vazirani2012} with devices $A$ and $B$, using $S_1$ as seed randomness, and store Protocol B's output in register $Y$.
		\item If Protocol B aborts, then abort $\VV$. Otherwise, continue.
		\item Output $\QExt_{n,r,\eps}(Y,S_2)$, where $n = n(\lfloor s/2 \rfloor)$, $r = v(s)$, and $\eps = 1/h(\lfloor s/2 \rfloor)$. 
\end{enumerate}
\end{minipage}
}
\end{center}
\caption{The $\VV$ sub-protocol. The functions $n(s)$ and $h(s)$ denote the output length and min-entropy lower bound of Protocol B in Theorem~\ref{thm:vv} on $s$ bits of seed.}
\end{figure}

\subsection{The $\RUV$ sub-protocol}

The $\RUV$ sub-protocol, using a random seed $S$, has two devices (call them $A$ and $B$) play a number $N$ of sequential CHSH games, where $N$ is a function of $|S|$, and the inputs to the devices in each of the CHSH games are determined by half of $S$. The $\RUV$ sub-protocol aborts if they do not win nearly $\approx \cos^2(\pi/8)$ fraction of games. Then, the other half of $S$ is used to select a random \emph{sub-block} of $A$'s outputs in the $N$ CHSH games, and the sub-block is produced as the output of $\RUV$. 

More precisely, let $X \in \{0,1\}^N$ denote $A$'s outputs. Divide $X$ into blocks of $t$ consecutive bits, and further subdivide each block into $\sqrt{t}$ sub-blocks of $\sqrt{t}$ bits each. We set $t = \lfloor N^{1/\alpha} \rfloor$, where $\alpha := \lceil 16 \kappa^2_* \rceil$ and $\kappa_*$ is the constant from~\cite[Theorem 5.7]{Reichardt2012}.

For all $s$, the $\RUV$ sub-protocol takes in a $s$-bit seed $S$, and outputs $r(s)$ bits, where $r(s) := \lfloor (s/4)^{1/(2\alpha)} \rfloor$.


\begin{figure}[H]
\begin{center}
\textbf{$\mathsf{RUV}(A, B ,S)$ Sub-Protocol}\\
\medskip
\framebox{
\begin{minipage}{0.9\textwidth}
\textbf{Input Non-signaling Devices}: $A,B$ \\
\textbf{Input Seed }: $S$ 

\begin{enumerate}
                \item Let $S_1$ be the first $\left \lfloor s/2 \right \rfloor$ bits of $S$, and $S_2$ be the last $\left \lfloor s/2 \right \rfloor$ bits of $S$, where $s := |S|$.
                \item Let $a,b \in \{0,1\}^{\lfloor s/4 \rfloor}$ be the first and last halves of $S_1$, respectively.
                \item For $i = 1,\ldots, N$, where $N := \lfloor s/4 \rfloor$:
                \begin{enumerate}
                		\item Input $a_i$, $b_i$ to devices $A$ and $B$ respectively, and collect outputs $x_i, y_i\in \{0,1\}$ from $A$ and $B$ respectively.
                \end{enumerate}
                \item Let $W$ be the number of indices $i$ such that $x_i \oplus y_i = a_i \wedge b_i$. If 
                $$
                		W < \cos^2 (\pi/8) N - \frac{1}{2\sqrt{2}} \sqrt{ N \log N },
                $$
                then abort $\RUV$. Otherwise, continue.
		
		\item Output $Z$, the $\sqrt{t}$-bit string that is the $j$th sub-block of the $i$th block of $X$, where $X$ is the register that holds the outputs $(x_i)$, and $i$ and $j$ are selected uniformly from $[N/t]$, $[\sqrt{t}]$, respectively, using the seed $S_2$.
\end{enumerate}
\end{minipage}
}
\end{center}
\end{figure}


%
%
%
%
%
%
%
%
%

\begin{figure}[H]
\begin{center}
\begin{tikzpicture}[scale=1,every node/.style={scale=1}]

\draw (0, -1) rectangle  (6.5,4);

\draw (1,3) rectangle node[font=\tiny] (S1) {$S_1$} (3,3.3);
\draw (3,3) rectangle node[font=\tiny] (S2) {$S_2$} (5,3.3);
\draw [decorate,decoration={brace,mirror, amplitude=5pt}] (1,2.8) -- (3,2.8);
\draw [decorate,decoration={brace,mirror, amplitude=5pt}] (3,2.8) -- (5,2.8);

\node[draw, minimum width=3cm,minimum height=1cm] (CHSHN) at (3, 1.5) {$\CHSHN$};
\draw (1,-0.3) rectangle node[font=\tiny] (X1) {} (5,0);
\draw (3.4, -0.3) rectangle node[font=\tiny] (X2) {$X_{ij}$} (4.4, 0);
\draw [decorate,decoration={brace, mirror, amplitude=4pt}] (3.4 , -0.4) -- (4.4,-0.4);

\node(Z) at (3.9, -2) {$Z$};

\draw [->] (3, 5) -- node[right,pos=0.3] {$S$} (3,3.4);
\draw [->] (2,2.7) .. controls +(down:0.5cm) .. (CHSHN);
\draw [->] (4,2.7) .. controls (5.5, 1.25) and (4.5, 0.8) .. node[pos=0.8,font=\tiny,right]{Sub-block selector} (X2);

\draw [->] (CHSHN) -- node[left] {$X$} (X1);
\draw [->] (3.9, -0.5) -- (Z);

\end{tikzpicture}
\end{center}
\caption{The $\RUV$ sub-protocol. All arrows indicate classical operations performed by the referee. }
\label{fig:ruv}
\end{figure}
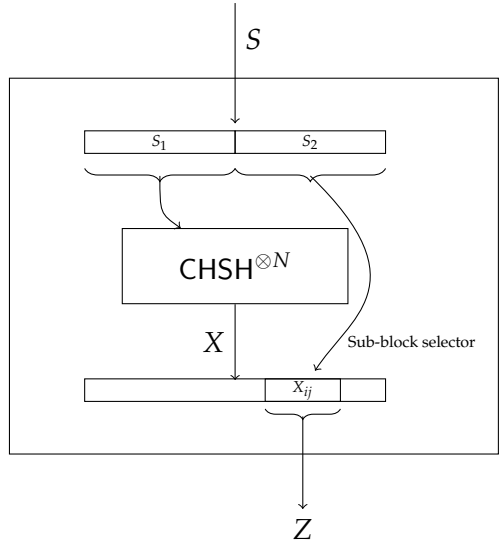

\section{Analysis of the $\InfiniteExp$ Protocol}  \label{sec:proof}
We now analyze the $\InfiniteExp$ protocol. As discussed in the Preliminaries (Section~\ref{sec:prelim}), we will use the notation $\rho^i$ and $\rho^f$ (or some variant thereof) to denote the state of the registers, devices, eavesdroppers, \etc, before and after the execution of a protocol, respectively. We will use the following functions throughout this section: $v(s)$ and $r(s)$ to denote the output lengths of the $\VV$ and $\RUV$ sub-protocols on inputs of length $s$, respectively (defined in Section~\ref{sec:proto}). The output length of the $\ClusterExp$ sub-protocol on an $s$-bit seed is $g(s) := r(v(s))$. We will use $g^{(k)}(s)$ to denote the $k$-fold composition of $g(s)$ (\ie $g^{(1)}(s) = g(s)$, $g^{(2)}(s) = g(g(s))$, \etc).

Theorem~\ref{thm:infinite-exp-completeness} establishes that there exists a quantum strategy by which the devices, with high probability, do not abort the $\InfiniteExp$ protocol.  Theorem~\ref{thm:infinite-exp} establishes the soundness of the $\InfiniteExp$ protocol.

\begin{theorem}[Completeness of the $\InfiniteExp$ protocol]
\label{thm:infinite-exp-completeness}
	There exists a non-signalling quantum strategy for devices $D_1,\ldots,D_8$, such that the probability that the referee aborts in any round $i$ in the execution of the $\InfiniteExp(C_1,C_2,S)$ protocol is at most $\exp(-\Omega(m^{1/3}))$, where $C_1 = \{D_1,\ldots,D_4\}$ and $C_2 = \{D_5,\ldots,D_8\}$, and $S$ is a uniformly random $m$-bit seed that is secure against $D_1,\ldots,D_8$.
\end{theorem}
\begin{proof}
	We group the devices into pairs $\{D_1,D_2\}$, $\{D_3,D_4\}$, $\{D_5,D_6\}$, and $\{D_7,D_8\}$, where pairs $\{D_1,D_2\}$ and $\{D_5,D_6\}$ will instantiate the ideal devices for the $\VV$ protocol (see~\cite{Vazirani2012} for more details), and the pairs $\{D_3,D_4\}$ and $\{D_7,D_8\}$ will instantiate the ideal devices for the $\RUV$ protocol (\ie use the ideal CHSH strategy in every round). 
	Fix a round $i$ and assume, without loss of generality, that the referee interacts with the pairs $\{D_1,D_2\}$ (used for the $\VV$ protocol) and $\{D_3,D_4\}$ (used for the $\RUV$ protocol) in round $i$. The probability that $\{D_1,D_2\}$ abort the $\VV$ protocol is at most $\exp(-\Omega(m_i^{2/3}))$, and the probability that $\{D_3,D_4\}$ abort the $\RUV$ protocol is at most $\exp(-\Omega(m_i^{1/3}))$, where $m_i = g^{(i)}(m)$. Thus, by the union bound, the probability of aborting any round $i$ is at most $\exp(-\Omega(m^{1/3}))$. 
\end{proof}

\begin{theorem}[Soundness of the $\InfiniteExp$ protocol]
\label{thm:infinite-exp}
	Let $C_0$ and $C_1$ be non-signaling Expansion Clusters. Suppose that a classical referee executes the $\InfiniteExp(C_0,C_1,S)$ protocol, where $S$ denotes the referee's classical register that holds an $m$-bit seed. Let $\WIN_i$ to be the event that the referee did not abort the $\InfiniteExp$ protocol in the $i$th round, and let $\WIN_{\leq i} = \WIN_1 \wedge \cdots \wedge \WIN_i$. Let $E$ be an arbitrary quantum system that may be entangled with $C_0$ and $C_1$, but cannot communicate with $C_0$ and $C_1$ once the protocol has started. Let $\rho^0_{SC_0C_1}$ denote the initial joint state of the seed and the clusters. If $\rho_{SC_0C_1} = U_m \otimes \rho_{C_0C_1}$, then we have for all $k \in \N$ that if $\Pr(\WIN_{\leq k}) \geq \lambda \geq \exp(-C'm^{1/3})$ for some universal constant $C'$, then
	$$
		\| \rho^k_{X_k E} - U_{g^{(k)}(m)} \otimes \rho^k_E \|_{\tr} \leq 4\exp(-C'' m^{1/3})/\lambda^2,
	$$
	where 
	\begin{itemize}
		\item $C''$ is the universal constant from Theorem~\ref{thm:cluster-exp}, and
		\item $\rho^k_{X_k E}$ denotes the joint state of the referee's $X_k$ register and $E$ after $k$ rounds of the $\InfiniteExp(C_0,C_1)$ Protocol, conditioned on the event $\WIN_{\leq k}$.
	\end{itemize}
\end{theorem}

Before presenting the proof of Theorem~\ref{thm:infinite-exp}, we wish to direct the reader's attention to the Input Security of the $\InfiniteExp$ protocol: the assumption on the initial seed is that it is in tensor product with the cluster devices only, and not the eavesdropper $E$ -- however, the output at each iteration is close to being in tensor product with the eavesdropper $E$.

The proof of Theorem~\ref{thm:infinite-exp} assumes the correctness of the $\ClusterExp$ sub-protocol (Theorem~\ref{thm:cluster-exp}), and shows that the $\InfiniteExp$ protocol maintains the property that at each iteration $i$, the output of $X$ of cluster $C_i$ (where $C_i$ denotes Expansion Cluster $C_{i \mod 2}$) is approximately secure against the other cluster $C_{i + 1}$. Thus, the the execution of the $\ClusterExp$ sub-protocol with $C_{i+1}$, conditioned on not aborting, will continue to produce a nearly uniform output. Furthermore, the errors accumulate linearly with each iteration.

\begin{proof}
Define $C_j := C_{j \mod 2}$. Divide the overall probability of success, $\Pr(\WIN_{\leq k})$, into conditional probabilities: let $p = \Pr(\WIN_{\leq k})$ and let $p_i = \Pr(\WIN_i | \WIN_{\leq i-1})$. Observe that we have $p = \prod p_i \geq \lambda$. We prove the claim by induction. 

\textbf{The inductive hypothesis}: Recursively define $\delta(i) := \eps_{\EC}(g^{(i-1)}(m),p_i) + \delta(i-1)/p_i$, where $\delta(1) := \eps_{\EC}(m,p_1)$ and $\eps_{\EC}(\cdot)$ is the error bound given by Theorem~\ref{thm:cluster-exp}. For all $i = 1,\ldots,k-1$, there exists a state $\mu^i_{XC_iC_{i+1}E}$ such that $\mu^i_{X_i C_{i+1} E} = U_{g^{(i)}(m)} \otimes \mu^i_{C_{i+1} E}$ and 
$$
	\| \rho^i_{X_i C_i C_{i+1} E} - \mu^i_{X_i C_iC_{i+1}E} \|_\tr \leq \delta(i),
$$
where $\rho^i_{X_iC_i C_{i+1}E}$ is the joint state of the $X_i$ register, both clusters $C_i$ and $C_{i+1}$, and $E$ after the $i$th round, conditioned on $\WIN_{\leq i}$.

Let $k=1$. Then, by invoking Theorem~\ref{thm:cluster-exp} with $C = C_1$, and treating the quantum eavesdropper as $C_2$ and $E$ together, we obtain that there exists a state $\mu^1_{X_1 C_1 C_2E}$ such that $\mu^1_{X_1 C_2 E} = U_{g(m)} \otimes \mu^1_{C_2E} $, and 
$$
	\| \rho^1_{X_1C_1 C_2 E} - \mu^1_{X_1 C_1 C_2 E} \|_\tr \leq \eps_{\EC}(m,p_1) = \delta(1).
$$ 
This establishes the base case. 

Now, suppose that we have run $k-1$ rounds of the $\InfiniteExp$ protocol for some $k > 1$. Using our inductive assumption for $i = k-1$, we invoke Theorem~\ref{thm:cluster-exp} along with Lemma~\ref{lem:approx} to conclude that there exists a state $\mu^k_{X_k C_k C_{k+1} E}$ such that $\mu^k_{X_k C_{k+1} E} = U_{g^{(k)}}(m) \otimes \mu^k_{C_{k+1} E}$ and
$$
	\| \rho^k_{X_k C_k C_{k+1} E} - \mu^k_{X_k C_k C_{k+1} E} \|_\tr \leq \eps_{\EC}(g^{(k-1)}(m),p_k) + \delta(k-1)/p_k := \delta(k).
$$
This completes the induction argument. We now bound $\delta(k)$:
\begin{align*}
	\delta(k) &= \eps_k + \frac{1}{p_k} \left (\eps_{k-1} + \frac{1}{p_{k-1}} \left( \eps_{k-2} + \cdots \right)  \right)  \\
	&\leq \frac{1}{\lambda} \left( \eps_k + \eps_{k-1} + \cdots + \eps_1 \right) \\
	&\leq \frac{2\eps_1}{\lambda},
\end{align*}
where we write $\eps_i := \eps_{\EC}(g^{(i)}(m),p_i)$, and use the facts that $\prod p_i \geq \lambda$ and each $\eps_i$ is exponentially smaller than $\eps_{i-1}$.

Finally, for every $k$, we have that 
\begin{align*}
	\| \rho^k_{X_k E} - U_{g^{(k)}(m)} \otimes \rho^k_E \|_\tr &\leq 
		\| \rho^k_{X_k E} - \mu^k_{X_k E} \|_\tr  + \| \mu^k_{X_k E} - U_{g^{(k)}(m)} \otimes \rho^k_E \|_\tr \\
		&\leq \delta(k) + \| U_{g^{(k)}}(m) \otimes \mu^k_{E} - U_{g^{(k)}(m)} \otimes \rho^k_E \|_\tr \\
		&= \delta(k) + \| \mu^k_{E} - \rho^k_E \|_\tr \\
		&\leq 2\delta(k).
\end{align*}

\end{proof}

Next, we argue that the $\ClusterExp$ sub-protocol is an Input Secure randomness expansion scheme. The correctness of the $\ClusterExp$ sub-protocol assumes the correctness of $\VV$ and $\RUV$ protocols (Lemmas~\ref{lem:vv-expander} and~\ref{lem:ruv-scrambler}, respectively). 

\begin{theorem}
\label{thm:cluster-exp}
	Let $C$ be an Expansion Cluster. Suppose that a classical referee executes the $\ClusterExp(C,S)$ protocol, where $S$ denotes the referee's classical register that holds an $m$-bit seed. Let $E$ be an arbitrary quantum system that may be entangled with $C$, but cannot communicate with $C$ once the protocol has started. If $\rho^i_{SC} = U_m \otimes \rho^i_C$, and $\Pr(\text{$\ClusterExp(C,S)$ succeeds}) \geq \lambda \geq \exp(-C'm^{1/3})$ for some universal constant $C'$, then there exists a state $\tau_{XCE}$ such that $\tau_{XE} = U_{g(m)} \otimes \tau_{E}$ and
	$$
		\| \rho^f_{XCE} - \tau_{XCE} \|_\tr \leq \eps_{\EC}(m,\lambda),
	$$
	where $\eps_{\EC}(m,\lambda) := \exp(-C'' m^{1/3})/\lambda$ for some universal constant $C''$, and $\rho^f_{XCE}$ is the joint state of the protocol's output $X$, the cluster $C$, and $E$ conditioned on the protocol $\ClusterExp(C,S)$ succeeding.
\end{theorem}
\begin{proof}
	Let $\lambda_1$ denote the probability that Step 1 of $\ClusterExp(C,S)$ succeeds, and let $\lambda_2$ denote the probability that Step 2 of $\ClusterExp(C,S)$ succeeds, conditioned on Step 1 succeeding, so that $\lambda_1\lambda_2 \geq \lambda$. Let $C$ consist of devices $D = \{D_1,D_2\}$ and $G = \{G_1,G_2\}$, where the $D_i$'s are used for execution of the $\VV$ protocol, and the $G_j$'s are used for the execution of the $\RUV$ protocol. Let $Y$ be the output of $\VV(D_1,D_2,S)$ (which is Step 1 of $\ClusterExp(C,S)$). By definition of the $\VV$ protocol, $|Y| = v(m)$. By Lemma~\ref{lem:vv-expander} and our assumption on $S$ (in particular, that $\rho^i_{SDG} = U_m \otimes \rho^i_{DG}$), there exists a state $\tau^v_{YDGE}$ such that $\tau^v_{YG} = U_{v(m)} \otimes \tau^v_G$ and
        \begin{equation}
	\label{eq:vv-step}
		\| \rho^v_{YDGE} - \tau^v_{YDGE} \|_{\tr} \leq \eps_{\VV}(m),
	\end{equation}
	where $\rho^v$ denotes the state of the system after running the $\VV$ protocol (and conditioned on it succeeding) but before executing the $\RUV$ protocol, and $\eps_{\VV}(\cdot)$ is the error bound given by Lemma~\ref{lem:vv-expander}.
	Let $X$ be the output of $\RUV(G_1,G_2,Y)$ (which is Step 2 of $\ClusterExp(C,S)$). By definition of the $\RUV$ protocol, $|X| = r(|Y|) = r(v(m))$. 
	
	Imagine that we executed the $\RUV$ protocol on the ``ideal'' input $\tau^v_{YDGE}$. By Lemma~\ref{lem:ruv-scrambler}, we would get that there existed a state $\tau^f_{XDGE}$ such that $\tau^f_{XE} = U_{g(m)} \otimes \tau^f_E$, and
	$$
		\| \rho^f_{XDGE} - \tau^f_{XDGE} \|_{\tr} \leq \eps_{\RUV}(v(m),\lambda_2),
	$$
	where $\eps_{\RUV}(\cdot,\cdot)$ is the error bound given by Lemma~\ref{lem:ruv-scrambler}. However, we only have the approximate guarantee on $Y$ given by~\eqref{eq:vv-step}. So, by Lemma~\ref{lem:approx}, we instead get that there exists a state $\tau^f_{XDGE}$ such that $\tau^f_{XE} = U_{g(m)} \otimes \tau^f_E$, and
	$$
		\| \rho^f_{XDGE} - \tau^f_{XDGE} \|_{\tr} \leq \eps_{\RUV}(v(m),\lambda_2) + \frac{ \eps_{\VV}(m)}{\lambda_2}.
	$$
	Plugging in the expressions for $\eps_{\RUV}$ and $\eps_{\VV}$, we get that this is at most 
	\begin{align*}
	\frac{1}{\lambda_2} (\sqrt{192(v(m)/4)^{-1/(8\alpha)}} + \sqrt{3 \exp(-C' m^{1/3})}) &\leq \exp(-C'' m^{1/3})/\lambda,
	\end{align*}
	for some universal constant $C''$. 
\end{proof}

\subsection{Analysis of the $\VV$ protocol}
In the next two sections, we analyze that the $\VV$ and the $\RUV$ components of the $\ClusterExp$ sub-protocol. As discussed in the introduction, the $\VV$ protocol in a cluster $C$ will provide near-exponential randomness expansion, although the analysis of~\cite{Vazirani2012} does not allow us to conclude that the output is secure against the other cluster $C'$ (\ie the Input Security Problem)~\footnote{See \ref{sec:related} for more about this issue.}. The $\RUV$ protocol in $C$ will be used to transform the output of $\VV$ to be secure against $C'$. Observe that, qualitatively, the $\RUV$ protocol solves the Input Security Problem because in Lemma~\ref{lem:ruv-scrambler}, the random seed is not required to be secure against an eavesdropper, yet the output is guaranteed to be! On the other hand, Lemma~\ref{lem:vv-expander} below requires the assumption that the seed to the $\VV$ protocol is secure against the protocol's devices and the eavesdropper simultaneously.

\begin{lemma}
\label{lem:vv-expander}
	Let $D_1,D_2$ be non-signaling quantum devices. Suppose that a classical referee executes the $\VV(D_1,D_2,S)$ protocol, where $S$ denotes the referee's classical register that holds an $m$-bit seed. Let $E$ be an arbitrary quantum system that may be entangled with $D_1$ and $D_2$, but cannot communicate with them once the protocol begins. If the initial joint state of $S$, $D_1$, $D_2$, and $E$ is $\rho^0_{SD_1D_2E} = U_m \otimes \rho^0_{D_1D_2E}$, and if $\Pr(\text{$\VV(D_1,D_2,S)$ succeeds}) \geq \exp(-C' m^{1/3})$ for some universal constant $C'$, then there exists a state $\tau_{XDE}$ where $\tau_{XE} = U_{v(m)} \otimes \rho^f_E$ and
	$$
		\| \rho^f_{XDE} - \tau_{XDE} \|_\tr \leq \eps_{\VV}(m),
	$$
	where $\rho^f_{XDE}$ is the joint state of $E$, the devices $D = \{D_1,D_2\}$, and the output $X$ of the protocol conditioned on the $\VV(D_1,D_2,S)$ protocol succeeding, $\eps_{\VV}(m) = \sqrt{3\exp(-C' m^{1/3})}$, and $v(m) = \exp(C' m^{1/3})/2$.
\end{lemma}
\begin{proof}
	The $\VV$ protocol consists of two parts, executing Protocol B of~\cite{Vazirani2012} using half of the seed $S$ (which we denote by $S_1$) to produce an output $Y$ of length $\exp(\Omega(m^{1/3}))$ which contains high min-entropy (conditioned on Protocol B not aborting), and then applying a randomness extractor using $Y$ as the source, and the other half of $S$ (which we denote by $S_2$) as the extractor seed, to produce an output $X$ that is close to uniform.
	
	Let $\rho^v_{YE}$ denote the joint state of the output of Protocol B (Step 2 of the $\VV$ protocol) and the eavesdropper $E$, conditioned on Protocol B not aborting. Then, by our assumption on $S$ and by Theorem~\ref{thm:vv}, we get that $H^{\eps}_{\infty}(Y  |  E)_{\rho^v} \geq h(m)$, where $h(m) = \exp(C' m^{1/3})$ and $\eps = \eps(m) = 1/h(m)$ for a universal constant $C'$.
	
	The $\VV$ protocol then applies a quantum-secure randomness extractor to the source $Y$, with seed $S_2$. The protocol uses the $\mathsf{QExt}: \{0,1\}^{|Y|} \times \{0,1\}^{d(m)} \to \{0,1\}^{h(m)/2}$ randomness extractor promised by Theorem~\ref{thm:dpvr}, where $d(m) =  \Theta(m)$. Let $\tilde{\rho}_{YE}$ be a cq-state that is $\eps$-close to $\rho^v_{YE}$ in trace distance, and is such that $H_{\infty}(Y | E)_{\tilde{\rho}} \geq h(m)$\, \footnote{Although the definition of smoothed min-entropy quantifies over \emph{all} density states in the $\eps$-ball around $\rho_{YE}$, there exists a cq-state with high min-entropy in the $\eps$-ball -- see, \eg,~\cite{renner2008security}.}. Then, since $\mathsf{QExt}$ is a $(h(m),\eps)$-quantum-secure extractor, we have that
	\begin{equation}
	\label{eq:extractor}
		\| \tilde{\rho}_{XE} - U_{v(m)} \otimes \tilde{\rho}_E \|_\tr \leq \eps,
	\end{equation}
	where $\tilde{\rho}_{XE}$ is the joint state of the output $X$ of the extractor $\mathsf{QExt}$ and $E$, with $\tilde{\rho}_Y$ as the source. View the application of $\mathsf{QExt}$ to the $Y$ and $S_2$ register as a trace-preserving quantum operation $\mathcal{E}$, which takes states $\rho^v_{YS_2}$ and outputs states $\rho^f_{\mathsf{QExt}(Y,S_2)}$. Then, by the triangle inequality, we have
	\begin{align*}
		\| \mathcal{E} \otimes \ident_E (\rho^v_{YS_2E}) - U_{v(m)} \otimes \rho^f_E \|_\tr \leq &\| \mathcal{E} \otimes \ident_E(\rho^v_{YS_2E}) - \mathcal{E} \otimes \ident_E(\tilde{\rho}_{YS_2E}) \|_\tr + \\
		 &\| \mathcal{E} \otimes \ident_E (\tilde{\rho}_{YS_2E}) - U_{v(m)} \otimes \tilde{\rho}_E \|_\tr + \\
		&\| U_{v(m)} \otimes \tilde{\rho}_E - U_{v(m)} \otimes \rho^f_E \|_\tr.
	\end{align*}
	Since $\mathcal{E}$ is trace-preserving, we can bound the first term by $\eps$. The second term is bounded by $\eps$ via equation~\eqref{eq:extractor}. The third term is bounded by $\eps$ because the trace distance is non-increasing with respect to the partial trace. Thus, 
	$$
		\| \rho^f_{XE} - U_{v(m)} \otimes \rho^f_E \|_\tr = \| \mathcal{E} \otimes \ident_E (\rho^v_{YS_2E}) - U_{v(m)} \otimes \rho^f_E \|_\tr \leq 3\eps.
	$$
	We then apply Lemma~\ref{lem:fidelity_trick} to obtain that there exists a state $\tau_{XDE}$ such that $\tau_{XE} = U_{v(m)} \otimes \rho^f_E$ and 
	$$
		\| \rho^f_{XDE} - \tau_{XDE} \|_\tr \leq \sqrt{3\eps}.	
	$$
	which proves the claim.
\end{proof}

\subsection{Analysis of the $\RUV$ protocol}

In this section, we analyze the $\RUV$ protocol. Before stating Lemma~\ref{lem:ruv-scrambler}, it will be necessary to give formal and precise definitions of several (classical) random variables, and how they interact with the relevant quantum states.

Let $S$ be an $m$-bit seed used in the $\RUV$ protocol, performed with non-signaling devices $D_1$ and $D_2$. Half of $S$, call it $S_1$, is used for $N$ CHSH games, where $N = m/4$. Recall that we divide the $N$ CHSH games into blocks of $t = N^{1/\alpha}$ consecutive games. Define the following random variables:
\begin{enumerate}
	\item Let $F$ denote the indicator variable that is $1$ iff the $\RUV$ protocol doesn't abort in Step 4 (\ie the devices win $\approx \cos^2(\pi/8) N$ CHSH games). Note that $F$ is a deterministic function of the devices' outputs and $S_1$.
	\item For all $i \in [N/t]$, let $I_i$ denote the indicator variable that is $1$ iff the devices $D_1$ and $D_2$ used a $\zeta$-ideal strategy to produce their outputs in the $i$th block of CHSH games, where $\zeta := \kappa_* t^{ -\kappa_*}$ (see Section~\ref{sec:prelim} and~\cite{Reichardt2012} for more details about ideal strategies).
	\item Let $H$ denote the indicator variable that is $1$ iff $G \geq (1 - \nu) N/t$, where $G := \sum I_i$ and $\nu := (12/\sqrt{2}) \sqrt{\log N} t/N^{1/4} \leq t^{-\alpha /8}$. 
\end{enumerate}

In our proof of Claim~\ref{lem:ruv-scrambler}, we will consider states such as $\rho_{FI_iXDE}$, where $X$ denotes the output of device $D_1$ after $N$ CHSH games, $D$ denotes the devices $D_1$ and $D_2$ together, $E$ denotes an arbitrary quantum system, $F$ will contain the classical bit indicating whether the devices aborted the $\RUV$ protocol or not, and $I_i$ will contain a classical bit denoting whether the devices used a $\zeta$-ideal strategy for block $i$. Because $F$ and $I_i$ are classical variables, $\rho_{FI_iXDE}$ is a cccqq-state, and thus there is an ensemble $\{\rho_{DE}^{fqx}\}$ that represents the states of the $D$ and $E$ systems conditioned on the classical events $F = f$, $I_i = q$, and $X = x$, where
$$
	\rho_{FI_iXDE} := \sum_{f,q,x} \Pr(F=f,I_i=q,X=x) \ketbra{f}{f}_F \otimes \ketbra{q}{q}_{I_i} \otimes \ketbra{x}{x}_X \otimes \rho_{DE}^{fqx}.
$$

Thus, we can meaningfully condition the state $\rho_{FI_iXDE} $ on various values of $F$ and $I_i$. For example, when we refer to the state $\rho_{XE | F=1}$, we mean the state that is, up to a normalization factor,
$$
	\sum_{q} \Pr(F=1,I_i=q,X=x) \ketbra{x}{x}_X \otimes \rho_{DE}^{1qx}.
$$
In particular, we will make use of the fact that $\rho_{XE | F = 1} = \Pr(I_i = 0 | F=1) \rho_{XE | I_i = 0, F = 1} +\Pr(I_i = 1 | F=1)  \rho_{XE | I_i = 1, F = 1}$, where $\rho_{XE | I_i =q,F=1}$ is defined similarly to $\rho_{XE|F=1}$.

\begin{lemma}
\label{lem:ruv-scrambler}
	Let $D_1,D_2$ be non-signaling quantum devices. Suppose that a classical referee executes the $\RUV(D_1,D_2,S)$ protocol, where $S$ denotes the referee's classical register that holds an $m$-bit seed. Let $E$ be an arbitrary quantum system that may be entangled with $D_1$ and $D_2$, but cannot communicate with them once the protocol begins.  If the initial joint state of $S$, $D_1$, and $D_2$ is $\rho^i_{SD_1D_2} = U_m \otimes \rho^i_{D_1D_2}$, and $\Pr(\text{$\RUV(D_1,D_2,S)$ succeeds}) \geq \lambda$, then, we have that there exists a state $\tau_{ZDE}$ where $\tau_{ZE} = U_{r(m)} \otimes \tau_E$, and 
	$$
		\| \rho^f_{ZDE | F = 1} - \tau_{ZDE} \|_\tr \leq \eps_{\RUV}(m,\lambda),
	$$
	where $\eps_{\RUV}(m,\lambda) \leq \sqrt{192(m/4)^{-1/(8\alpha)}/\lambda}$, and where $\rho^f_{ZDE|F=1}$ is the joint state of $E$, the devices $D = \{D_1,D_2\}$, and the output $Z$ of the protocol, conditioned on $F = 1$ (\ie the $\RUV(D_1,D_2,S)$ protocol does not abort).
\end{lemma}

\begin{proof}
Let $\rho^i_{XDFE}$ be the joint state of the $X$, $D$, $F$, and $E$ registers \emph{before} the $N$ CHSH games are played (so $X$ and $F$ are initialized to the all $0$ state). For this proof, we will assume that $E$ is such that $\rho^i_{XDFE}$ is a \emph{pure state}. This is without loss of generality, because we can take a non-pure state $\rho^i_{XDFE}$ and augment it with some extension $E' \supset E$ such that $\rho^i_{XDFE'}$ is pure (\eg via a purification of the state $\rho^i_{XDFE}$). Observe that $\| \rho^f_{ZE' | F = 1} - U_{r(m)} \otimes \rho^f_{E' | F = 1} \|_\tr \leq \eps$ implies $\| \rho^f_{ZE | F = 1} - U_{r(m)} \otimes \rho^f_{E | F = 1} \|_\tr \leq \eps$, because the trace distance is non-increasing under discarding the augmented system $E' \backslash E$.

For notational clarity, we shall omit the superscripts $i$ and $f$, because we focus on the state $\rho_{FSXDE}$ of the system \emph{after} the $N$ CHSH games (\ie the $X$ register holds the output of device $D_1$), but \emph{before} conditioning on $F=1$ and before using the seed $S_2$ to select a sub-block.  The $i^{th}$ block of $X$ will be denoted $X_i$, and the $j^{th}$ sub-block of the $i^{th}$ block will be denoted $X_{ij}$.

There are two main components to this proof.
\begin{enumerate}
	\item  We argue that, for the state $\rho_{XE | F = 1}$, there is a $1 - \delta$ fraction of sub-blocks $X_{ij}$ such that 
$$\|  \rho_{X_{ij}E |  F =1} - U_{\sqrt{t}} \otimes \rho_{E |  F =1} \|_\tr \leq \eta,$$
where we set $\eta$ and $\delta$ later in the proof. We say that such sub-blocks are $\eta$-good with respect to $E$.  
\item We argue that the string $S_2$  (substring of the seed $S$ used to select the sub-block that $\RUV(D_1,D_2,S)$ will output) is in tensor product with a string describing the locations of the $\eta$-good sub-blocks of the state $\rho_{XE | F = 1}$.    
\end{enumerate}

In particular, let $Z := X_{S_2}$ denote the sub-block selected by string $S_2$. From the above two components, it follows that, for the state $\rho_{XE | F = 1}$, the the random variable $Z$ is $(\eta + \delta)$-good with respect to $E$, \ie,
$$
	\|  \rho_{ZE |  F =1} - U_{\sqrt{t}} \otimes \rho_{E |  F =1} \|_\tr \leq \eta + \delta.
$$
We then invoke Lemma~\ref{lem:fidelity_trick} to argue that there exists a state $\tau_{ZDE}$ such that $\tau_{ZE} = U_{\sqrt{t}} \otimes \rho_{E | F=1}$ and 
$$
	\|  \rho_{ZDE |  F =1} - \tau_{ZDE} \|_\tr \leq \sqrt{\eta + \delta},
$$
and we are done. We now proceed to proving the first two components.

\textbf{There are many good sub-blocks.} By the definition of $I_i$ and Lemma~\ref{lem:ideal_strategy},
$$ 
	\|\rho_{X_iE | I_i = 1} - U_{t} \otimes \rho_{E | I_i = 1}   \|_\tr \leq \zeta .
$$
It follows by Proposition~\ref{chainrule} that, for at least a $1-t^{-1/4}$ fraction of sub-blocks $j$ of block $i$ we have that 
$$
	\| \rho_{X_{ij}EF | I_i = 1} - U_{\sqrt{t}} \otimes \rho_{EF | I_i = 1}  \|_\tr  \leq \mu,
$$
where $\mu := 2(\sqrt{\zeta} + t^{-1/8})$. If we then condition on the event $F = 1$ it follows that
\begin{align} \label{goodsubblocks}
\| \rho_{X_{ij}E | I_i = 1, F =1} - U_{\sqrt{t}} \otimes \rho_{E | I_i = 1, F =1}  \|_\tr  \leq \frac{\mu}{\Pr(F = 1)} \leq    \frac{\mu}{\lambda}
\end{align}

We wish to establish the above statement for the state $\rho_{X_{ij}E|F = 1}$ rather than the state $\rho_{X_{ij}E| I_i = 1, F = 1}$.  The key to making this transition is to establish that, for many values of $i$, the event $F = 1$ is approximately a sub-event of the event $I_i = 1$.  To do so, it is helpful to consider the event $H =1$.  

Let $M := N/t$ denote the number of blocks of CHSH games. It follows from the definition of $H$ that $\sum_{i \in [M]} \E [ I_i = 0 | H = 1] \leq \nu M$.  Thus, by Markov's inequality we have that at most a $\sqrt{\nu}$ fraction of blocks $i \in [M]$ are such that $\Pr( I_i = 0 | H = 1) > \sqrt{\nu}$.  Thus, at least a $1 - \sqrt{\nu}$ fraction of blocks $i \in [M]$ have $\Pr( I_i = 0 | H = 1) \leq \sqrt{\nu}$.  

Consider such a block $i$.  Note that by Theorem~\ref{thm:ruv_sequential_rigidity}, $\Pr( H = 0, F = 1) \leq t^{-2}$. Thus
\begin{align*}
\Pr(I_i = 0 , F = 1) &= \Pr(I_i = 0 | H=1, F=1) \Pr(H =1, F=1) + \Pr(I_i = 0 | H = 0, F=1) \Pr(H = 0, F = 1) \\
& \leq \Pr(I_i = 0 | H=1, F=1)  + \Pr(I_i = 0 | H = 0, F=1) t^{-2} \\
   &\leq \frac{\Pr(I_i = 0 | H = 1)}{\Pr(F = 1)} +  t^{-2} \\
	& \leq \frac{\sqrt{\nu}}{\lambda} + t^{-2}.
\end{align*}

Since $I_i = 1$ is a classical event, we have $\rho_{XE|F=1} = (1 - \tau) \rho_{XE|I_i = 1, F=1} + \tau \rho_{XE|I_i = 0, F=1}$, where $\tau := \Pr(I_i  = 0 | F=1)$. 
Thus,
	\begin{align*}
		\| \rho_{X_iE|F=1} - \rho_{X_iE|I_i = 1, F=1}  \|_{\tr}  &= \| (-\tau) \rho_{X_iE|I_i = 1, F=1} + \tau \rho_{X_iE|I_i = 0, F=1} \|_\tr \\
							&\leq \tau ( \|\rho_{X_iE|I_i = 1, F = 1} \|_\tr + \| \rho_{X_iE | I_i = 0, F=1}\|_\tr) \\
							&\leq 2\tau.
	\end{align*}
	By definition, $\tau = \Pr(I_i = 0,F=1)/\Pr(F=1)$. Thus, 
$$
	\| \rho_{X_{i}E | I_i = 1, F =1}  - \rho_{X_{i}E |  F =1} \|_\tr \leq 2\frac{\sqrt{\nu} + \lambda t^{-2}}{\lambda^2}
$$
By tracing over all except the $j^{th}$ sub-block we get 
\begin{align}
\| \rho_{X_{ij}E | I_i = 1, F =1}  - \rho_{X_{ij}E |  F =1} \|_\tr \leq 2\frac{\sqrt{\nu} + \lambda t^{-2}}{\lambda^2} \label{FsubIforsubblock}
\end{align}
By tracing over the entire $X_i$ register we get 
\begin{align}
\| \rho_{E | I_i = 1, F =1}  - \rho_{E |  F =1} \|_\tr \leq 2\frac{\sqrt{\nu} + \lambda t^{-2}}{\lambda^2} \label{FsubIforE}
\end{align}
Thus, at least a $(1 - t^{-1/4})(1-\sqrt{\nu})$ of all the sub-blocks $X_{ij}$ have the property that equations \eqref{goodsubblocks}, \eqref{FsubIforE}, and \eqref{FsubIforsubblock} all hold.  It follows by the triangle inequality that
\begin{align}
\|  \rho_{X_{ij}E |  F =1} - U_{\sqrt{t}} \otimes \rho_{E |  F =1} \|_\tr &\leq  \| \rho_{X_{ij}E |  F =1} - \rho_{X_{ij}E |  I_i = 1, F =1} \|_\tr+ \| \rho_{X_{ij}E |  I_i = 1, F =1} -  U_{\sqrt{t}} \otimes \rho_{E | I_i=1, F =1} \|_\tr \nonumber \\
& + \| U_{\sqrt{t}} \otimes \rho_{E | I_i=1, F =1} - U_{\sqrt{t}} \otimes \rho_{E | F =1}  \|_\tr  \nonumber \\ 
&\leq 2\frac{\sqrt{\nu} + \lambda t^{-2}}{\lambda^2} + \frac{\mu}{\lambda} + \| \rho_{E | I_i=1, F =1} -  \rho_{E | F =1}  \|_\tr \nonumber \\
& \leq 4 \left ( \frac{\sqrt{\nu} + \lambda t^{-2}}{\lambda^2}\right ) + \frac{\mu}{\lambda} \nonumber \\
&\leq \frac{96}{\lambda} t^{-1/8}.
\end{align}

Define $\eta := 96t^{-1/8}/\lambda$.  
Thus, we have that at least a $1-\delta$ fraction of the sub-blocks $X_{ij}$ are $\eta$-good with respect to $E$, where $\delta := t^{-1/4} + \sqrt{\nu} \leq 2t^{-1/4}$.  It is easy to see that $\eta + \delta \leq 2\eta = 192(m/4)^{-1/(8\alpha)}/\lambda$.

\textbf{$S_2$ is secure against the location of good sub-blocks.}
Although we have established that most of the sub-blocks of $X$ are $\eta$-good, we need to show that the seed $S_2$ used to select the sub-block for the output of the $\RUV$ protocol is independent of the locations of the good sub-blocks (i.e. the indices $i, j$ such that $X_{ij}$ is $\eta$-good with respect to $E$). \emph{A priori}, since $S_2$ is entangled with the eavesdropper $E$ (because $S_2$ was the output of a different expansion cluster), it could be that $S_2$ was somehow adversarially generated to select a bad sub-block. Here, we show that this cannot happen, because the locations of the good sub-blocks can be \emph{locally computed} by the devices $D = \{D_1, D_2\}$. Since $\rho^i_{SD} = U_m \otimes \rho^i_D$ (where $\rho^i_D := \rho_{D_1D_2}$), $S_2$ is independent of the good sub-block locations.

Consider the following thought experiment: the system $D = \{D_1,D_2\}$ is augmented with a \emph{classical description} $\Delta$ of the state $\rho^i_{XFD}$, and a register $\Lambda$ that will store the locally computed location of the good sub-blocks, so that we have a new system $D' = \{D_1,D_2,\Delta,\Lambda \}$. Throughout the $\RUV$ protocol, the $D'$ system cannot communicate with the eavesdropper system $E$. At the beginning of the $\RUV$ protocol, we have that $\rho_{SD'} = U_{|S|} \otimes \rho_{D'}$. 
Imagine that we have measured the $S_1$ register (but the $S_2$ register remains unmeasured), so that it is now a deterministic value $s_1$. Let $\mathcal{E}_{s_1}$ denote the quantum operation that acts on the systems $D_1$, $D_2$, $F$ that represents the strategy employed by devices $D_1$ and $D_2$, on the inputs determined by $s_1$, for the $N$ CHSH games (Step 3 of the $\RUV$ protocol). That is, $\rho^f_{XFD} := \mathcal{E}_{s_1}(\rho^i_{XFD})$.

As part of this thought experiment, we imagine that, after the $N$ CHSH games, the $\Delta$ system performs a quantum operation $\mathcal{S}_{s_1}$ on the $\Delta$, and $\Lambda$ systems (but not $D_1$ and $D_2$!) to classically simulate the strategy used by the devices $D_1, D_2$ on input $s_1$ in the $N$ CHSH games, and compute the location of the good sub-blocks. The $\Delta$ will then contain a classical description of the state $\rho^f_{XFD}$. Note that at this point, $S_2$ is still secure against $D'$; that is, we have
$$
	\mathcal{S}_{s_1}(\mathcal{E}_{s_1}(\rho^i_{S_2 XFD\Delta \Lambda})) = U_{|S_2|} \otimes \mathcal{S}_{s_1}(\mathcal{E}_{s_1}(\rho^i_{XFD\Delta \Lambda})).
$$

We elaborate on the classical simulation $\mathcal{S}$. Given the classical description $\Delta$ of $\rho^i_{XFD}$, the location of the good sub-blocks can be computed by using $\Delta$ in the following way:
\begin{enumerate}
	\item Compute the classical description of a purification $\sigma^i_{XFDE'}$ of the state $\rho^i_{XFD}$. Note that in general, $\sigma^i_{XFDE'}$ is different from the ``real'' state $\rho^i_{XFDE}$ , because the $\Delta$ system has no knowledge of the external system $E$.
	\item Classically simulate the devices' strategy $\mathcal{E}$ on the state $\sigma^i_{XFDE'}$, \ie,
	$$
		\sigma^f_{XFDE'} = \mathcal{E}_{s_1}(\sigma^i_{XFDE'}).
	$$
	Note that $\sigma^f_{XFD} = \rho^f_{XFD}$.
	\item Compute the indices $i$, $j$, such that
	$$
		\| \sigma^f_{X_{ij} E' | F = 1} - U_{\sqrt{t}} \otimes \sigma^f_{E' | F =1} \|_{\tr} \leq \eta,
	$$
	and store those indices in a register $\Lambda$.
\end{enumerate}

We now argue that $\Lambda$ will contain an accurate description of the locations of the $\eta$-good sub-blocks in the ``real'' state $\rho^f_{XFDE}$. From this, since $\rho^f_{S_2 \Lambda} = \rho^f_{S_2} \otimes \rho^f_{\Lambda}$, it follows that $S_2$ is independent of the good sub-block locations.

Here we will use the assumption, stated at the beginning of this proof, that $\rho^i_{XFDE}$ is a pure state. Let $\rho^i_{XFDE} := \ketbra{\psi}{\psi}$, and let $\sigma^i_{XFDE'} := \ketbra{\phi}{\phi}$. There exists a unitary $V$ that takes the $E$ system to the $E'$ system and acts as the identity on all other systems, such that $\ket{\phi} = V\ket{\psi}$. Since $V$ and $\mathcal{E}_{s_1}$ act on different systems, they commute, and hence $\sigma^f_{XFDE'} = V \rho^f_{XFDE} V^{\dagger}$. Furthermore, $V$ commutes with the projector $\Pi_{F = 1}$ that projects onto the $F = 1$ subspace, and thus
$$
	\sigma^f_{XDE' | F = 1} = V \rho^f_{XDE | F=1} V^{\dagger}.
$$
Thus,
\begin{align*}
	\| \sigma^f_{X_{ij} E' | F = 1} - U_{\sqrt{t}} \otimes \sigma^f_{E' | F =1} \|_{\tr} &= 
			\| \tr_{\neq (i,j), D} ( V \rho^f_{XDE | F=1} V^{\dagger}) - U_{\sqrt{t}} \otimes \tr_{XD}(V \rho^f_{XDE | F=1} V^{\dagger}) \|_{\tr} \\
			&= \| V \left ( \tr_{\neq (i,j), D} (\rho^f_{XDE | F=1}) - U_{\sqrt{t}} \otimes \tr_{XD}(\rho^f_{XDE | F=1}) \right) V^{\dagger} \|_{\tr} \\
			&= \| \tr_{\neq (i,j), D} (\rho^f_{XDE | F=1}) - U_{\sqrt{t}} \otimes \tr_{XD}(\rho^f_{XDE | F=1}) \|_{\tr} \\
			&= \| \rho^f_{X_{ij} E | F = 1} - U_{\sqrt{t}} \otimes \rho^f_{E | F =1} \|_{\tr},
\end{align*}
where $\tr_{\neq (i,j), D}$ indicates tracing out over all sub-blocks except for the $j$th one in the $i$th block, and the system $D$. The second equality follows from the fact that $V$ and the partial trace commute. The third equality follows because the trace norm is unitarily invariant. 

Thus, the indices $i,j$ where $\| \sigma^f_{X_{ij} E' | F = 1} - U_{\sqrt{t}} \otimes \sigma^f_{E' | F =1} \|_{\tr} \leq \eta$ are exactly those sub-blocks that are $\eta$-good in the state $\rho^f_{XFDE}$. 

\end{proof}

\begin{proposition} \label{chainrule}
Let $i \in [N/t]$ be the index of a block. If
$$ \|\rho_{X_iE | I_i = 1} - U_{t} \otimes \rho_{E | I_i = 1}   \| \leq \zeta,  $$
then for at least a $1-t^{-1/4}$ fraction of sub-blocks $j$ of block $i$ we have that 
$$\| \rho_{X_{ij}EF | I_i = 1} - U_{\sqrt{t}} \otimes \rho_{EF | I_i = 1}  \|  \leq 2 (\sqrt{\zeta} +  t^{-1/8}).$$
\end{proposition}

\begin{proof}
	By Lemma~\ref{lem:fidelity_trick}, there exists a state $\sigma_{X_i F E}$ such that $\sigma_{X_i E} = U_t \otimes \rho_{E|I_i = 1}$, and $\| \rho_{X_i FE | I_i = 1} - \sigma_{X_i FE} \|_\tr \leq \sqrt{\zeta}$. Let $R := \sqrt{t}$ denote the number of sub-blocks in a block. We now prove the Proposition by showing that, for the state $\sigma_{X_iFE}$, at least $1 - t^{-1/4}$ fraction of sub-block indices $j \in [R]$ satisfy $I(X_{ij} : FE)_{\sigma} \leq 2t^{-1/4}$. For such $j$, we obtain:
	\begin{align*}
		\| \rho_{X_{ij}FE|I_i =1} - U_{\sqrt{t}} \otimes \rho_{FE|I_i =1} \|_{\tr} &\leq  \| \rho_{X_{ij}FE|I_i =1} -  \sigma_{X_{ij}FE}\|_{\tr} + \| \sigma_{X_{ij}FE} - U_{\sqrt{t}} \otimes \sigma_{FE} \|_{\tr} \\
                &  + \| U_{\sqrt{t}} \otimes \sigma_{FE} - U_{\sqrt{t}} \otimes \rho_{FE|I_i =1}  \|_{\tr} \\
                 &\leq \sqrt{\zeta} + \sqrt{4t^{-1/4}} + \sqrt{\zeta}.
	\end{align*}
	The bound on the second term in the second inequality is given via Pinsker's Inequality (see Fact~\ref{fact:qinfo}), which states that $\| \sigma_{X_{ij}FE} - U_{\sqrt{t}} \otimes \sigma_{FE} \|_{\tr} \leq \sqrt{2 I(X_{ij} : FE)_{\sigma} }$. The bounds on the first and third terms come from the fact that the trace distance is non-increasing with respect to the partial trace.

	Thus we focus on analyzing the state $\sigma_{X_i FE}$ for the remainder of this proof. We apply the chain rule to obtain $I(X_i : FE)_{\sigma} = \sum_j I(X_{ij} : FE | X_{i, <j})_{\sigma}$. This is equivalent to $$\E_j[ I(X_{ij} : FE | X_{i, <j})_{\sigma} ] = \frac{1}{R} I(X_i : FE)_{\sigma},$$ where $X_{i,<j}$ denotes all the $X_{ik}$ such that $k < j$. We will omit the subscript $\sigma$ because the underlying state is clear from context.
	We upper-bound the quantity $I(X_i : FE)$ via the following calculation:
	\begin{align}
		\label{eq:mut_1}
		I(X_i : FE) &= H(X_i) - H(X_i | FE) \\
		\label{eq:mut_2}
				 &= H(X_i) - (H(X_i FE) - H(FE)) \\
		\label{eq:mut_3}
				 &= H(X_i) - (H(X_i FE) - H(E) - H(F|E)) \\
		\label{eq:mut_4}
				 &= H(X_i) - (H(X_i E) + H(F | X_i E) - H(E) - H(F|E))\\
		\label{eq:mut_5}
				 &= H(X_i) - (H(X_i) + H(E) + H(F | X_i E) - H(E) - H(F|E)) \\
		\label{eq:mut_6}
				 &= H(F|E) - H(F|X_i E) \\
		\label{eq:mut_7}
				 &\leq 2H(F) \\
				 &\leq 2
	\end{align}
	Equation~\eqref{eq:mut_1} is the definition of mutual information. Equations~\eqref{eq:mut_2},~\eqref{eq:mut_3}, and~\eqref{eq:mut_4} follow from the definition of conditional mutual entropy. Equation~\eqref{eq:mut_5} follows from our assumption that $\sigma_{X_i E} = \sigma_{X_i} \otimes \sigma_{E}$. Equation~\eqref{eq:mut_7} follows from the fact that conditioning can only reduce entropy, and that $-H(F|X_iE) \leq H(F)$.
	
	We now lower bound the individual terms of the expectation $I(X_{ij} : FE | X_{i,<j})$. 
	\begin{align}
		\label{eq:lb_1}
		I(X_{ij} : FE | X_{i,<j}) &= H(X_{ij} | X_{i,<j}) - H(X_{ij} | FEX_{i,<j}) \\
		\label{eq:lb_2}
						 &\geq H(X_{ij}) - H(X_{ij} | FE) \\
		\label{eq:lb_3}				 
						 &= I(X_{ij} : FE).
	\end{align}
	Equation~\eqref{eq:lb_1} is the definition of conditional mutual information. Equation~\eqref{eq:lb_2} follows because $\sigma_{X_i} = U_{t}$ (hence $\sigma_{X_{ij}}$ is in tensor product with $\sigma_{X_{i,<j}}$), and conditioning can only reduce entropy. Finally, equation~\eqref{eq:lb_3} is again the definition of mutual information.
	
Thus, $$\E_j[ I(X_{ij} : FE) ] \leq \frac{2}{R},$$ and by Markov's inequality, we get that $1 - \mu$ fraction of $j$'s are such that $I(X_{ij} : FE) \leq \frac{2}{\mu R}$. Setting $\mu = t^{-1/4}$ completes the proof.
\end{proof}

\section{Conclusion}
We have presented a randomness expansion protocol that achieves infinite expansion: starting with $m$ bits of uniform seed, the protocol produces an arbitrarily long output string that is $\exp(-\Omega(m^{\frac{1}{3}}))$-close to uniform. Furthermore, this protocol only requires eight non-signaling quantum devices (and can be performed with just six devices using a simple modification).  In order to accomplish this we design an \emph{Input Secure} adaptive randomness expansion protocol, which is then used as a sub-protocol in the infinite expansion protocol.  We suspect that the existence of Input Secure randomness expansion protocols is also of independent interest.  As evidence of their independent interest we note that Input Secure protocols play a key role as a building block in \cite{csw14}, where they were discovered independently from this work, and used to design a protocol for seedless randomness amplification from any min-entropy source (see Section~\ref{sec:related}).  


\textbf{Acknowledgments}. We thank Aram Harrow, Ben Reichardt, Xiaodi Wu, Kai-Min Chung, Fernando Brand\~{a}o, and especially Thomas Vidick for stimulating discussions about randomness expansion.

\appendix
\section{Proof of Lemma~\ref{lem:approx}}
\label{app:cond}
\begin{proof}[Proof of Lemma~\ref{lem:approx}]
Define $\mu_{XDE}$ to be the state $\tau_{XDE}$ as given by the assumption in the lemma on input $\sigma_{FSXDE}$ where $\sigma_{FSXD} = \sigma_{FSX} \otimes \sigma_{D}$. By the triangle inequality, we have:
\begin{align}
\label{em:lem-approx-tri}
		\| \rho^f_{XDE} - \mu_{XDE} \|_{\tr} \leq 
			&\| \mathcal{FE} \otimes \ident_E(\rho^i_{FSXDE}) - \mathcal{FE} \otimes \ident_E(\sigma_{FSXDE}) \|_{\tr} \\
			&+ \| \mathcal{FE} \otimes \ident_E( \sigma_{FSXDE} ) - \mu_{XDE} \|_{\tr} \nonumber.
\end{align}
We bound the first term on the right hand side:
\begin{align*}
	\| \mathcal{FE} \otimes \ident_E (\rho^i_{FSXDE}) - \mathcal{FE} \otimes \ident_E(\sigma_{FSXDE} ) \|_{\tr} &\leq \frac{1}{\lambda} \| \mathcal{E} \otimes \ident_E (\rho^i_{FSXDE}) - \mathcal{E}\otimes \ident_E(\sigma_{FSXDE}) \|_{\tr} \\
			&\leq \frac{1}{\lambda} \|\rho^i_{FSXDE} - \sigma_{FSXDE} \|_{\tr} \\
			&\leq \delta/\lambda.
\end{align*}
Let $\lambda'$ denote the probability that the $F$ register of the state $\mathcal{E} \otimes \ident_E (\sigma_{FSXDE})$, when measured, has outcome $\ket{1}$. Note that $\max\{\lambda,\lambda'\} \geq \lambda$, so the first inequality follows from Lemma~\ref{lem:conditioning}. The second inequality follows because trace-preserving quantum operations are contractive with respect to the trace distance. The final inequality comes from our assumption on $\rho^i_{FSXDE}$.

The second term on the right hand side of~\eqref{em:lem-approx-tri} is bounded by $\eps$ from our assumption on the quantum operation $\mathcal{FE}$. 
\end{proof}

\section{Useful lemmata}
\begin{lemma}
\label{lem:conditioning}
Let $\rho_{FQ},\sigma_{FQ}$ be cq-states on the same classical-quantum Hilbert space $\Hilb_F \otimes \Hilb_Q$. Let $E$ be a set of outcomes of the $F$ register such that $\min\{\Pr_\rho(E),\Pr_\sigma(E)\} > 0$, where $\Pr_\rho(E)$, $\Pr_\sigma(E)$ denote the probabilities of obtaining outcome $E$ when measuring $\rho_F$ and $\sigma_F$ in the computational basis. Then,
$$
	\| \rho_{FQ|E} - \sigma_{FQ|E} \|_{\tr} \leq \frac{ \| \rho_{FQ} - \sigma_{FQ} \|_{\tr}}{\max\{\Pr_\rho(E),\Pr_\sigma(E)\}},
$$
where $\rho_{FQ|E}$ and $\sigma_{FQ|E}$ denote the post-measurement state of $\rho_{FQ}$ and $\sigma_{FQ}$, respectively, conditioned on $E$.
\end{lemma}
\begin{proof}

We use the operational interpretation of the trace norm of two quantum states, namely, that $\| \rho - \sigma \|_{\tr} = \max_{A} \Pr(A(\rho) = 1) - \Pr(A(\sigma) = 1)$, where $\rho$ and $\sigma$ are arbitrary density matrices, and the maximization is over all possible 0/1-valued POVMs $A$.

Let $\lambda_\rho$ and $\lambda_\sigma$ denote $\Pr_\rho(E)$ and $\Pr_\sigma(E)$ respectively. We consider two cases: $\lambda_\rho \geq \lambda_\sigma$ and $\lambda_\rho < \lambda_\sigma$. Take the first case.

Consider the following two-outcome experiment $A$ that tries to distinguish between $\rho_{FQ}$ and $\sigma_{FQ}$. We first measure the $F$ register in the computational basis. If the outcome $E$ does not occur, we output ``0''. Suppose outcome $E$ does occur. Let $B$ be the optimal two-outcome POVM such that $\Pr(B(\rho_{FQ|E}) = 1) - \Pr(B(\sigma_{FQ|E}) = 1) = \| \rho_{FQ|E} - \sigma_{FQ|E} \|_{\tr}$. We then make the measurement dictated by $B$ on the post-measurement state (which is either $\rho_{FQ|E}$ or $\sigma_{FQ|E}$), and output ``1'' iff $B$ outputs $``1''$. Then, we have that
\begin{align*}
		\| \rho_{FQ} - \sigma_{FQ} \|_{\tr} &\geq \Pr(A(\rho_{FQ}) = 1) - \Pr(A(\sigma_{FQ}) = 1) \\
								   &= \lambda_\rho \Pr(B(\rho_{FQ|E}) = 1) - \lambda_\sigma \Pr(B(\sigma_{FQ|E}) = 1) \\
								   &= \lambda_\rho \left ( \| \rho_{FQ|E} - \sigma_{FQ|E} \|_{\tr} + \Pr(B(\sigma_{FQ|E}) = 1) \right) - \lambda_\sigma \Pr(B(\sigma_{FQ|E}) = 1).
\end{align*}	
Solving for $ \| \rho_{FQ|E} - \sigma_{FQ|E} \|_{\tr} $, we get that
$$
	 \| \rho_{FQ|E} - \sigma_{FQ|E} \|_{\tr} \leq \frac{\| \rho_{FQ} - \sigma_{FQ} \|_{\tr} - \beta(\lambda_\rho - \lambda_\sigma)}{\lambda_\rho} \leq \frac{\| \rho_{FQ} - \sigma_{FQ} \|_{\tr}}{\lambda_\rho} \leq \frac{\| \rho_{FQ} - \sigma_{FQ} \|_{\tr}}{\max\{\lambda_\rho,\lambda_\sigma\}},
$$
where $\beta := \Pr(B(\sigma_{FQ|E}) = 1)$. In the other case, we have that $\lambda_\rho < \lambda_\sigma$. We can then switch the order of $\rho_{FQ}$ and $\sigma_{FQ}$ in the previous argument, and obtain that 
$$
	 \| \rho_{FQ|E} - \sigma_{FQ|E} \|_{\tr} \leq \frac{\| \rho_{FQ} - \sigma_{FQ} \|_{\tr}}{\lambda_\sigma} \leq \frac{\| \rho_{FQ} - \sigma_{FQ} \|_{\tr}}{\max\{\lambda_\rho,\lambda_\sigma\}}.
$$
\end{proof}

\begin{lemma}
\label{lem:fidelity_trick}
	Let $\rho_{A_1A_2B} \in \mathcal{D}(\Hilb_{A_1} \otimes \Hilb_{A_2} \otimes \Hilb_B)$, and $\sigma_{A_1A_2} \in \mathcal{D}(\Hilb_{A_1} \otimes \Hilb_{A_2})$ be such that $\rho_{A_1A_2B}$ is a cqq-state, $\sigma_{A_1A_2}$ is a cq-state, and $\| \rho_{A_1A_2} - \sigma_{A_1A_2} \|_\tr \leq \eps$. Then there exists a cqq-state $\tau_{A_1A_2B} \in \mathcal{D}(\Hilb_{A_1} \otimes \Hilb_{A_2} \otimes \Hilb_B)$ such that $\tau_{A_1A_2} = \sigma_{A_1A_2}$ and $\| \rho_{A_1A_2B} - \tau_{A_1A_2B} \|_\tr \leq \sqrt{\eps}$.
\end{lemma}
\begin{proof}
	For notational brevity we will let $A = \{A_1,A_2\}$ so $\rho_{AB} := \rho_{A_1A_2B}$ and $\sigma_A := \sigma_{A_1A_2}$. Let $\Fidelity(\rho,\sigma)$ denote the fidelity between two quantum states $\rho$ and $\sigma$. By Uhlmann's Theorem, there exists purifications $\rho_{AQ} := \ketbra{\psi}{\psi}$ and $\sigma_{AQ} := \ketbra{\phi}{\phi}$ of $\rho_A$ and $\sigma_A$, respectively, such that $ \Fidelity(\rho_{A},\sigma_{A}) = | \braket{\psi}{\phi} |$~\cite{wilde2013quantum}. But by the Fuchs-van de Graaf inequalities, we also have that $\Fidelity(\rho_A,\sigma_A) \geq 1 - \| \rho_A - \sigma_A \|_\tr/2 \geq 1 - \eps/2$~\cite{wilde2013quantum}. Since $\| \rho_{AQ} - \sigma_{AQ} \|_\tr = \sqrt{1 - | \braket{\psi}{\phi}|^2}$, we have that
	$$
		\| \rho_{AQ} - \sigma_{AQ} \|_\tr \leq \sqrt{\eps}.
	$$
	Let $\rho_{ABR} = \ketbra{\theta}{\theta}$ be a purification of the state $\rho_{AB}$. Since $\rho_{ABR}$ and $\rho_{AQ}$ are both purifications of the state $\rho_A$, there exists a unitary map $V$ that takes the $Q$ space to the $BR$ space such that $\rho_{ABR} = V\rho_{AQ}V^{\dagger}$. Define $\tau'_{ABR} := V\sigma_{AQ}V^{\dagger}$. Then, by the unitary invariance of the trace norm, we have that
	\begin{align*}
		\| \rho_{ABR} - \tau'_{ABR} \|_\tr &= \| V\rho_{AQ} V^\dagger - V\tau'_{AQ} V^\dagger \|_\tr \\
								&=\| V(\rho_{AQ}- \tau'_{AQ}) V^\dagger \|_\tr \\
								&=\| \rho_{AQ} - \tau'_{AQ} \|_\tr \\
								&\leq \sqrt{\eps}.
	\end{align*}
	Since the trace norm cannot increase when discarding subsystems, we obtain $\| \rho_{AB} - \tau'_{AB} \|_\tr \leq \sqrt{\eps}$. $\tau'_{AB} = \tau'_{A_1A_2B}$ is not guaranteed to be a cqq-state, but we can apply the trace-preserving quantum map $\mathcal{E}$ that measures the $A_1$ system in the computational basis and forgets the measurement outcome. Let $\tau_{A_1A_2B} := \mathcal{E}(\tau'_{A_1A_2B})$, and observe that this is a cqq-state. Since $\rho_{A_1A_2B}$ is already a cqq-state, $\rho_{A_1A_2B} = \mathcal{E}(\rho_{A_1A_2B})$. Because trace-preserving quantum maps are contractive under the trace norm, we obtain $\| \rho_{A_1A_2B} - \tau_{A_1A_2B} \|_\tr \leq \sqrt{\eps}$, and we are done.
\end{proof}

\section{Parameter settings for the $\VV$ sub-protocol}
\label{app:params}
For the sake of concreteness, we specify the settings of parameters to be used in the instantiation of Protocol B of~\cite{Vazirani2012} in our $\VV$ sub-protocol (see Section~\ref{sec:proto}). We choose constants $\alpha, \gamma > 0$ such that $\gamma \leq 1/(10+ 8 \alpha)$.  
These constants are part of the definition of $\mathsf{VV}$ and will remain unchanged for every instance of $\mathsf{VV}$ throughout the $\InfiniteExp$ protocol.
             
In~\cite[Theorem 2]{Vazirani2012}, the parameter $h$ specifies the min-entropy lower bound of Protocol B, which in turn governs the length of the seed to Protocol B and length of the output. By definition Protocol B implemented with parameter $h$ requires at most $K_1 \gamma^{-3} \log^3 (h)$ bits of seed for some fixed constant $K_1$ (this constant may depend on $\alpha$, but since $\alpha$ is a global constant here, we ignore this). When Protocol B is invoked by $\mathsf{VV}(A,B, S)$, we will set $h = \left \lfloor  2^{\gamma \left( \left \lfloor s/2 \right \rfloor \frac{1}{K_1} \right)^{1/3} } \right \rfloor $, where $s := |S|$, and it follows that Protocol B, with these parameters, will require no more than $\left \lfloor s/2 \right \rfloor $ bits of seed.

We will now discuss parameters relevant to the quantum extractor which will be used in $\mathsf{VV}$.  Let us now define $t := h^{\frac{1}{\gamma}}$, $C := \lceil 100 \alpha \rceil$ and $\eps := \frac{1}{h}$.  The output of Protocol B is a bit string of length $n := \lceil 10 \log^2(t) \rceil \cdot \lceil C t \log^2(t) \rceil$.  By Theorem \ref{thm::quantumextractor} there exists a function $\mathsf{QExt}:\{0,1\}^{n} \times \{0,1\}^{d} \to \{0,1\}^{\frac{h}{2}}$ that is a $(\frac{h}{2} + O(\log \left (\frac{h}{2} \right )) + O(\log 1/\eps), \eps)$-quantum-secure extractor as long as $d \geq O \left (\log^2(n/\eps) \log \left ( \frac{h}{2} \right) \right ) = O \left (\log^3 (h) \right ) = O \left ( \gamma^3 \left \lfloor s/2 \right \rfloor \frac{1}{K_1}  \right )$.  That is, as long as $d \geq K_4  \gamma^3 \left \lfloor s/2 \right \rfloor \frac{1}{K_1} $ for some fixed constant $K_4$.  

Thus, in specifying the $\VV$ sub-protocol and throughout the paper, we will set the following functions, where $s$ is the length of input to the $\VV$ sub-protocol:
\begin{itemize}
	\item Min-entropy lower bound of Protocol B:
	$$
		h(s) := \left \lfloor  2^{\gamma \left( \left \lfloor s/2 \right \rfloor \frac{1}{K_1} \right)^{1/3} } \right \rfloor.
	$$
	\item Output length of Protocol B:
	$$
		n(s) := \left \lfloor 10C \left( \frac{s}{2K_1} \right)^{4/3} 2^{ (s/(2K_1))^{1/3} } \right \rfloor.
	$$
	\item Seed length of the extractor:
	$$
		d(s) := \left \lceil \frac{K_4}{K_1}  \gamma^3 \left \lfloor s/2 \right \rfloor \right \rceil.
	$$ 
	\item Output length of the extractor/$\VV$ sub-protocol:
	$$
		v(s) := \lfloor h(s)/2 \rfloor.
	$$
\end{itemize}

\bibliographystyle{alphaabbrvprelim}
\bibliography{bibliography}

\end{document}